\documentclass{article}
\usepackage{graphicx} 

\usepackage[top=30mm, bottom=30mm, left=30mm, right=30mm]{geometry}
\usepackage{epsfig,hyperref,url}
\usepackage{amsmath}
\usepackage{amsfonts}
\usepackage{amsthm}
\usepackage{amssymb}
\usepackage{ascmac}
\usepackage{booktabs}
\usepackage{color}
\usepackage{graphicx}
\usepackage{bm}
\usepackage{multirow}
\usepackage{setspace}
\usepackage{natbib}
\usepackage{mathtools}
\mathtoolsset{showonlyrefs}
\usepackage{csquotes}
\usepackage{enumitem}
\setlist[enumerate]{wide=0pt, leftmargin=15pt, labelwidth=15pt, align=left}
\usepackage{mathrsfs}
\usepackage{comment}
\usepackage{here}
\usepackage{listings}
\theoremstyle{plain} 
\newtheorem{theorem}{Theorem}

\newtheorem{corollary}{Corollary}

\theoremstyle{definition}
\newtheorem{definition}{Definition}

\newtheorem{remark}{Remark}

\title{Modeling asymmetry in multi-way contingency tables \\ with ordinal categories via $f$-divergence}
\author{Hisaya Okahara \thanks{Email : hisaya.okahara@gmail.com. Postal address : Yamazaki, Noda, 2788510, Chiba, Japan.} \and Kouji Tahata  \and \\
Department of Information Sciences, Tokyo University of Science}
\date{\empty}

\begin{document}

\maketitle

\begin{abstract}
This study introduces a novel model that effectively captures asymmetric structures in multivariate contingency tables with ordinal categories. 
Leveraging the principle of maximum entropy, our approach employs $f$-divergence to provide a rational model under the presence of a ``prior guess.''
Inspired by the constraints used in the derivation of multivariate normal distributions, we demonstrate that the proposed model minimizes $f$-divergence from complete symmetry under specific constraints. 
The proposed model encompasses existing asymmetry models as special cases while offering remarkably high interpretability. 
By modifying divergence measures included in $f$-divergence, the model provides the flexibility to adapt to specific probabilistic structures of interest.
Furthermore, we established theorems that show that a complete symmetry model can be decomposed into two or more models, each imposing less restrictive parameter constraints. 
We also investigated the properties of the goodness-of-fit statistics with an emphasis on the likelihood ratio and Wald test statistics.
Extensive Monte Carlo simulations confirmed the nominal size, high power, and robustness of the choice of $f$-divergence.
Finally, an application to real-world data highlights the practical utility of the proposed model for analyzing asymmetric structures in ordinal contingency tables.
\end{abstract}

\section{Introduction\label{sec:intro}}
Complete symmetry is a fundamental concept in multivariate categorical data analysis.
Specifically, complete symmetry ensures that all random variables are exchangeable and that the effects of their categories are equivalent across all pairs.
This property is particularly relevant in applications such as public opinion surveys and mobility analyses, where transitions between categories (e.g., opinions and cities) are expected to exhibit symmetric behavior. 
For instance, in the mobility data, complete symmetry implies that the flow of individuals from cities A to B is equivalent to the flow in the opposite direction. 
However, in many real-world settings, such as changes in public opinion regarding government policies and shifts in regional migration patterns, this assumption is often violated.
These violations necessitate the relaxation of the \textit{complete symmetry} (S) model by adopting less restrictive alternatives that capture asymmetric structures.
Previous studies introduced asymmetric models to provide more flexible structures for categorical data \citep[see][]{bhapkar1990Marginal, kateri1997Asymmetry, kateri2021Families}.

Although previous studies have primarily focused on symmetric and asymmetric models, the theoretical connection between contingency tables and continuous distributions is of interest. 
\cite{pearson1900Criterion} advocated the idea of assuming that underlying continuous bivariate distributions underlie two-way contingency tables.
For ordinal two-way contingency tables, \cite{agresti1983Simple}, inspired by Pearson's assumption of an underlying continuous distribution, proposed a \textit{linear diagonals-parameter symmetry} (LS) model based on the structure of a bivariate normal distribution with equal marginal variances.
Following this, \cite{tomizawa1991Extended} extended the model to unequal marginal variances, denoting it as the \textit{extended} LS (ELS) model.
Subsequently, \cite{yamamoto2007Decomposition} proposed the \textit{generalized} LS (GLS) model, which can accommodate unequal correlation coefficients and is defined for multi-way contingency tables, thus providing greater flexibility for capturing complex dependencies in higher dimensions.

In this study, we revisit these models from the perspective of the maximum entropy principle \citep{jaynes1957Information} using a generalized divergence measure, $f$-divergence \citep{ali1966General, csiszar1967Informationtype}.
We demonstrate that these models are derived by maximizing the relative entropy \citep{kullback1951Information} subject to specific constraints.
Traditional entropy maximization, based on Kullback-Leibler (KL) divergence, leads to the well-known \textit{log-linear model} \citep[p.346]{bishop2007Discretea}, which serves as a fundamental framework for statistical modeling of categorical data \citep{agresti2012Categorical}. 
Contrarily, employing $f$-divergence offers greater flexibility by allowing the selection of divergence measures better suited to the specific characteristics of the data and the objectives of the analysis.

To address the violation of complete symmetry, we adopted the principle of maximum entropy using the S model as a reference model (referred to as a ``prior guess'').
Modeling strategies depend on the availability of prior information regarding explicit or latent ordinal structures.
A key distinction is whether categorical variables have an inherent order.
When categorical variables lack an inherent order, models typically rely on marginal probabilities \citep{kateri1997Asymmetry, kateri2021Families}.
However, when an ordinal structure is available, incorporating information from the marginal distribution, such as means and variances, is natural to achieve a balance between model flexibility and structural constraints \citep{kateri1994Fdivergence, kateri2007Class, okahara2025Generalized}.
In this study, we focus on the latter scenario by exploring maximum-entropy models that leverage available marginal information.
This formulation yields a simple yet plausible asymmetric model that satisfies specific linear constraints, aligning with the model-generation framework proposed by \citet[p.365]{bishop2007Discretea}.

The remainder of the study is organized as follows. 
Section \ref{sec:preliminaries} introduces the notation used throughout the study and the divergence measures. 
Section \ref{sec:model} proposes a generalized class of $f$-divergence-based models and discusses useful special cases of the proposed model, thereby providing a unified interpretation of the parameters and highlighting the consistent role of $f$-divergence.
Section \ref{sec:properties} provides the decompositions of the symmetric structure and investigates the asymptotic properties of several test statistics, which allows us to investigate the reasons for the poor fit of complete symmetry.
Section \ref{sec:simulation} presents a simulation that demonstrates the relationship with a multivariate normal distribution.
Section \ref{sec:example} illustrates the application to real-world data.
Finally, Section \ref{sec:conclusion} summarizes the study.

\section{Preliminaries}
\label{sec:preliminaries}
In this section, we introduce the notations used in this study and discuss the key divergence measures. 
Section \ref{subsec:notation} presents the fundamental notation for multivariate categorical data analysis. 
Section \ref{subsec:divergence} introduces important divergence measures included in $f$-divergence.
Additional notation is introduced in Section \ref{subsec:test}. 
See Table \ref{table:notation} in Appendix \ref{appendix-notation}.

\subsection{Notation} \label{subsec:notation}
Let $V = \{1, \dots, T\}$ be an index set, and let ($X_j,\ j \in V$) denote variables with $ X_j \in \mathcal{I} = \{1, \dots, r\} $, where $ r $ is the number of categories for each variable. 
We define the set of all possible category combinations as $\mathcal{I}_V = \mathcal{I} \times \cdots \times \mathcal{I} = \mathcal{I}^{|V|}$, where $|\cdot|$ denotes the number of elements in the set. 
The elements of $ \mathcal{I}_V $ are referred to as cells of the contingency table, and $r^T$ cells are present. 
A cell is denoted generically by $ \boldsymbol{i} $, where $ \boldsymbol{i} = (i_1, \dots, i_T) \in \mathcal{I}_V $.
The joint probability mass function of $(X_1, \dots, X_T)$ is denoted by $\{ \pi_{\boldsymbol{i}} \}_{\boldsymbol{i} \in \mathcal{I}_V}$ and defined as
\begin{equation}
    \pi_{\boldsymbol{i}} = \mathrm{Pr}(X_1 = i_1, \dots, X_T = i_T), \quad \boldsymbol{i} \in \mathcal{I}_{V},
\end{equation}
where $\pi_{\boldsymbol{i}} > 0$ and $\sum_{\boldsymbol{i} \in \mathcal{I}_V} \pi_{\boldsymbol{i}} = 1$.

In this study, as we aim to quantify the deviation from complete symmetry, we consider the S model as a ``prior guess.''
To explicitly describe the symmetric structures, we defined the following set:
\begin{align*}
    D(\boldsymbol{i}) &= D(i_1,\dots,i_T) \\
    &= \{ (j_1,\dots,j_T) \in \mathcal{I}_{V} \ | \ (j_1,\dots,j_T) \ \text{is any permutation of} \ (i_1,\dots,i_T) \}.
\end{align*}
Subsequently, the complete symmetric structure $\boldsymbol{\pi}^{S} = \{\pi_{\boldsymbol{i}}^{S} \}$ is denoted as
\begin{equation} \label{reference:symmetry}
    \pi_{\boldsymbol{i}}^{S} = \frac{1}{|D(\boldsymbol{i})|} \sum_{\boldsymbol{j} \in D(\boldsymbol{i})} \pi_{\boldsymbol{j}}, \quad \boldsymbol{i} \in \mathcal{I}_{V}.
\end{equation}
We define $\pi_{\boldsymbol{i}}^S$ as the probability under the complete symmetry model.

\subsection{Divergence Measures} \label{subsec:divergence}
Let $\boldsymbol{p} = (p_{\boldsymbol{i}})$ and $\boldsymbol{q} = (q_{\boldsymbol{i}})$ be two discrete finite multivariate probability distributions; then, the \textit{f}-divergence measure from $\boldsymbol{q}$ to $\boldsymbol{p}$ is defined by
\begin{equation} \label{f-divergence}
    D_f (\boldsymbol{p} \parallel \boldsymbol{q}) = \sum_{\boldsymbol{i} \in \mathcal{I}_{V}} q_{\boldsymbol{i}} f \left( \frac{p_{\boldsymbol{i}}}{q_{\boldsymbol{i}}} \right),
\end{equation}
where $f$ is a real-valued convex function on $(0,\infty)$ with $f(1)=0$, $f(0) = \lim_{x \rightarrow 0} f(x)$, $0 f (0/0) = 0$ and $0 f(a/0 ) = a\lim_{t\rightarrow \infty} [f(t)/t]$ for $a > 0$.
$f$-divergence was introduced by \cite{ali1966General} and \cite{csiszar1967Informationtype}.

$f$-divergence is the only class of decomposable information monotonic divergence.
A divergence measure is decomposable if it can be expressed as the sum of componentwise contributions, whereas information monotonicity ensures the monotonic behavior of the divergence under data processing, such as coarse graining \citep[see][]{amari2009Divergence, amari2010Information}.
These properties enable the structured decomposition of complete symmetry and conditional modeling of asymmetry, which are the key advantages of our modeling framework.

It is well known that the $f$-divergence remains unchanged under linear transformations of the form $\bar{f}(x) = f(x) + c(x-1)$ with a constant $c \in \mathbb{R}$.
Moreover, because for any constant $c>0$ we have
\begin{equation}
    D_{cf}(\boldsymbol{p} \parallel \boldsymbol{q}) = c D_{f}(\boldsymbol{p} \parallel \boldsymbol{q}).
\end{equation}
Constant $c$ determines the divergence scale.
Hence, without loss of generality, we can use a convex function satisfying $f(1)=0$, $f^{\prime}(1)=0$ and $f^{\prime\prime}(1)=1$, which is called a standard $f$-function.
An $f$-divergence derived from a standard $f$-function is referred to as a standard $f$-divergence \citep{amari2016Information}. 
In the following discussion, we adopt a standard $f$-divergence and denote the first derivative of the function $f$ by $F$ (i.e., $F = f^{\prime}$).

The $f$-divergence encompasses a broader class of divergence measures, including the power divergence family \citep{cressie1984Multinomial}, which includes the KL-divergence as a special case and provides a general framework that extends it.
This family is defined by the following settings:
\begin{equation} \label{eq-power-func}
 f_{\lambda}(x) = 
 \begin{cases}
      \frac{x(x^{\lambda}-1)}{\lambda(\lambda+1)} - \frac{x-1}{\lambda+1},
      & \text{if} \ \lambda \notin \{-1, 0\}, \\
      x - 1 - \log x,  
      & \text{if} \ \lambda = -1, \\
      1 - x + x\log x,  
      & \text{if} \ \lambda = 0,
    \end{cases}
\end{equation}
where $\lambda$ is a real-valued parameter.
\begin{equation} \label{power-divergence}
    D_{\lambda}(\boldsymbol{p} \parallel \boldsymbol{q}) 
    = \frac{1}{\lambda(\lambda + 1)} \sum_{\boldsymbol{i} \in \mathcal{I}_{V}} p_{\boldsymbol{i}} \biggl[ \biggl( \frac{p_{\boldsymbol{i}}}{q_{\boldsymbol{i}}} \biggr)^{\lambda} - 1 \biggr].
\end{equation}
For $\lambda=-1$ or $\lambda=0$, these cases can be defined by the continuous limits, and $D_{\lambda}(\boldsymbol{p} \parallel \boldsymbol{q})$ is continuous in $\lambda$.
These are expressed as follows:
\begin{equation*}
    \lim_{\lambda \rightarrow -1} D_{\lambda}(\boldsymbol{p} \parallel \boldsymbol{q}) = D_{KL}(\boldsymbol{q} \parallel \boldsymbol{p})
    \quad \text{and} \quad
    \lim_{\lambda \rightarrow 0} D_{\lambda}(\boldsymbol{p} \parallel \boldsymbol{q}) = D_{KL}(\boldsymbol{p} \parallel \boldsymbol{q}).
\end{equation*}
The power-divergence family, parameterized by $\lambda$, encompasses a wide range of well-known divergence measures as special cases. 
As demonstrated above, the limits $\lambda \rightarrow -1$ and $\lambda \rightarrow 0$ yield the reverse KL-divergence and KL-divergence, respectively.
Additionally, specific values of $\lambda$ correspond to other important divergence measures. 
For example, when $\lambda = -1/2$, the power divergence corresponds to the Hellinger distance \citep{hellinger1909Neue}, and when $\lambda = 1$, it yields Pearson's $\chi^2$-divergence \citep{pearson1900Criterion}.

\section{Modeling Framework based on \textit{f}-divergence} \label{sec:model}
Suppose that one should infer from insufficient information that specifies only a feasible set of statistical models.
The principle of maximum entropy is a powerful tool for solving such inverse problems.
In such a situation, a key idea is to use data to inform our beliefs, representing our uncertainty, which forms the basis for decision making given our ``prior guess.'' 
The maximum entropy principle offers a framework for addressing this challenge and is the only sound criterion for assigning probabilities in the absence of information.
This principle has deep roots in various fields such as statistics, information theory, and decision theory, as it provides a way to synthesize evidence while accounting for uncertainty.
For a detailed discussion, refer to \cite{jaynes1957Information, csiszar1975IDivergence}.

\subsection{Generalized Gaussian Symmetry Model} \label{subsec:gaussian}
To describe ordinal structures, we assign numerical scores $u_1 < \cdots < u_r$ to the $r$ ordered categories of each variable $(X_j,\ j\in V)$. 
These scores $\{u_k\}$ are treated as known constants that characterize the ordinal nature of the categories, rather than as unknown parameters to be estimated.
The corresponding score vector is denoted by $\boldsymbol{u}_{\boldsymbol{i}} = (u_{i_1}, \dots, u_{i_T})^{\top}$ for $\boldsymbol{i} \in \mathcal{I}_{V}$.

In this section, we explore the general framework of modeling based on $f$-divergence minimization. 
This approach generalizes the KL-divergence minimization, which is a limiting special case of the generalized framework.
We propose the new model in Definition \ref{defn-1}, which is referred to as the \textit{$f$-divergence-based Gaussian symmetry} (GS[$f$]) model.
\begin{definition} \label{defn-1}
    Let $f$ be a twice-differentiable and strictly convex function, and let $F(x) = f^{\prime}(x)$ for all $x$. 
    Given scores $\{u_k\}$ to each ordinal category, the \textit{$f$-divergence-based Gaussian symmetry} model is defined as
    \begin{equation} \label{defn-gs-f}
       \pi_{\boldsymbol{i}} = \pi_{\boldsymbol{i}}^{S} 
            F^{-1} \left( \boldsymbol{u}_{\boldsymbol{i}}^{\top} \boldsymbol{\alpha}
            + \boldsymbol{u}_{\boldsymbol{i}}^{\top} \boldsymbol{B}\,\boldsymbol{u}_{\boldsymbol{i}}
            + \gamma_{\boldsymbol{i}} \right), \quad \boldsymbol{i}\in\mathcal{I}_{V},
    \end{equation}
    where $\boldsymbol{\alpha}=(\alpha_1,\dots,\alpha_T)^{\top}$, $\boldsymbol{B}=(\beta_{st})_{s,t\in V}$ is a symmetric matrix, and $\gamma_{\boldsymbol{i}}=\gamma_{\boldsymbol{j}}$ for $\boldsymbol{j}\in D(\boldsymbol{i})$. 
    For identifiability, we set $\alpha_T=\beta_{TT}=\beta_{(T-1)T}=0$ without loss of generality. 
\end{definition}

\begin{theorem} \label{thm-1}
    In the class of models with given scores $\{u_i\}$, given the sums of probabilities over each symmetric set $\{\sum_{\boldsymbol{j} \in D(\boldsymbol{i})} \pi_{\boldsymbol{j}} \}$, given marginal means $\{ \sum_{\boldsymbol{i} \in \mathcal{I}_{V}} u_{i_{s}} \pi_{\boldsymbol{i}} \}$ and given second-order mixed moments $\{\sum_{\boldsymbol{i} \in \mathcal{I}_{V}} u_{i_{s}} u_{i_{t}} \pi_{\boldsymbol{i}} \}$ (which are equivalent to fixing the covariances once the means are specified) between the classification variables, the model in Definition \ref{defn-1} is the unique minimizer of the $f$-divergence from the S model.
\end{theorem}

\begin{proof}
The optimization problem is formulated as follows:
\begin{align}
    \text{Minimize}     &\qquad  D_{f}(\boldsymbol{\pi} \parallel \boldsymbol{\pi}^{S}) \\
    \text{subject to}   &\quad  \sum_{\boldsymbol{j} \in D(\boldsymbol{i})} \pi_{\boldsymbol{j}} = v_{\boldsymbol{i}}^{S}, \quad \boldsymbol{i} \in \mathcal{I}_{V} \\
                        &\quad  \sum_{\boldsymbol{i} \in \mathcal{I}_{V}} u_{i_{s}} \pi_{\boldsymbol{i}} = \mu_{s},
                        \quad  \sum_{\boldsymbol{i} \in \mathcal{I}_{V}} u_{i_{s}} u_{i_{t}} \pi_{\boldsymbol{i}} = \sigma_{st}, \quad s,t \in V.
\end{align}
Note that the superscript $S$ in $v_{\boldsymbol{i}}^S$ does not denote an exponent; rather, $v_{\boldsymbol{i}}^S$ represents the total probability of the symmetric set $D(\boldsymbol{i})$.
This optimization problem can be solved using Lagrange multipliers as follows:
\begin{align*}
    L(\{\pi_{\boldsymbol{i}} \}) &= D_{f}( \boldsymbol{\pi} \parallel \boldsymbol{\pi}^{S}) 
    + \sum_{s \in V} \lambda_{s} \left( \sum_{\boldsymbol{i} \in \mathcal{I}_{V}} u_{i_{s}} \pi_{\boldsymbol{i}} - \mu_{s} \right) \\
    &\qquad + \sum_{s,t \in V} \lambda_{st} \left( \sum_{\boldsymbol{i} \in \mathcal{I}_{V}} u_{i_{s}} u_{i_{t}} \pi_{\boldsymbol{i}} - \sigma_{st} \right) 
    + \sum_{\boldsymbol{i} \in \mathcal{I}_{V}} \eta_{\boldsymbol{i}} \left( \sum_{\boldsymbol{j} \in D(\boldsymbol{i})} \pi_{\boldsymbol{j}} 
    - v_{\boldsymbol{i}}^{S} \right).
\end{align*}
Equating the partial derivative of $L(\{\pi_{\boldsymbol{i}} \})$ to $0$ regarding $\pi_{\boldsymbol{i}}$ yields:
\begin{equation*}
    0 = 
    f^{\prime} \left( \frac{\pi_{\boldsymbol{i}}}{\pi_{\boldsymbol{i}}^{S}} \right) 
    + \sum_{s \in V} \lambda_{s} u_{i_{s}} 
    + \sum_{s,t \in V} \lambda_{st} u_{i_{s}} u_{i_{t}} 
    + \sum_{\boldsymbol{j} \in D(\boldsymbol{i})} \eta_{\boldsymbol{j}}.
\end{equation*}
Let $-\lambda_{s}$, $-(\lambda_{st} + \lambda_{ts})/2$ and $-\sum_{\boldsymbol{j} \in D(\boldsymbol{i})} \eta_{\boldsymbol{j}}$ denote $\alpha_{s}$, $\beta_{st}$ and $\gamma_{\boldsymbol{i}}$, respectively. 
We have
\begin{equation} \label{eq-F_pc}
    F \left( \frac{\pi_{\boldsymbol{i}} }{\pi_{\boldsymbol{i}}^{S}} \right) = 
    \sum_{s \in V} \alpha_{s} u_{i_{s}}
    + \sum_{s,t \in V} \beta_{st} u_{i_{s}} u_{i_{t}} 
    + \gamma_{\boldsymbol{i}},
\end{equation}
where $\beta_{st}=\beta_{ts}$ and $\gamma_{\boldsymbol{i}} = \gamma_{\boldsymbol{j}}$ for $\boldsymbol{j} \in D(\boldsymbol{i})$. 
When the $f$-function is strictly convex, it follows that $F'(x)=f''(x) >0$ for all $x$. 
Hence, $F$ is strictly monotonic, which ensures that $F^{-1}$ exists. 
From the above equation, we obtain
\begin{equation}
    \pi_{\boldsymbol{i}} = \pi_{\boldsymbol{i}}^{S} 
    F^{-1} \left(\boldsymbol{u}_{\boldsymbol{i}}^{\top} \boldsymbol{\alpha}
    + \boldsymbol{u}_{\boldsymbol{i}}^{\top} \boldsymbol{B} \boldsymbol{u}_{\boldsymbol{i}}
    + \gamma_{\boldsymbol{i}} \right), \quad \boldsymbol{i} \in \mathcal{I}_{V},
\end{equation}
where $\boldsymbol{\alpha}=(\alpha_{1}, \dots,\alpha_T)^{\top}$, $\boldsymbol{B}=(\beta_{st})_{s,t \in V}$ is a symmetric matrix, and $\gamma_{\boldsymbol{i}} = \gamma_{\boldsymbol{j}}$ for $\boldsymbol{j} \in D(\boldsymbol{i})$.
\qed
\end{proof}

\begin{remark}
In general, positivity of the estimated cell probabilities is not automatically guaranteed only by Definition \ref{defn-1}.
Although for some choices of $f$-function (e.g., the KL-divergence, where $F^{-1}(x)=e^x$) positivity follows directly, whereas in other cases (e.g., Pearson’s $\chi^2$ divergence, where $F^{-1}(x)=x+1$) the term $F^{-1}(\cdot)$ can yield negative values. 
Therefore, positivity must be enforced in the estimation procedure as part of the optimization problem, so that the estimated probabilities satisfy $\pi_{\boldsymbol{i}}>0$ and $\sum_{\boldsymbol{i}\in\mathcal{I}_V} \pi_{\boldsymbol{i}}=1$.
\end{remark}

For brevity, we denote the model \eqref{defn-gs-f} as the GS$[f]$ model throughout the study, where G stands for \textit{generalized} or \textit{Gaussian}, reflecting the structural resemblance to the multivariate normal distribution.
This nomenclature was motivated by Pearson’s foundational work, in which an underlying continuous distribution was assumed for latent variables.
In the continuous setting, the multivariate Gaussian distribution is known to maximize Shannon’s differential entropy under constraints on the mean and covariance matrix \citep[see][]{cover2012Elements}.
Similarly, in a discrete setting, where the distribution is supported on the lattice, these constraints lead to a lattice Gaussian distribution as the maximum-entropy distribution \citep[see][]{kemp1997Characterizationsa, nielsen2022Kullback}.
In the next section, we explore the specific interpretations of the proposed GS$[f]$ model, examine its behavior under special cases, and examine its relationship with existing models.

\subsection{Specific Choices of \textit{f}-functions} \label{subsec:specialcase}
Under the constraints $\pi_{\boldsymbol{i}}^{S} = \pi_{\boldsymbol{j}}^{S}$ and $\gamma_{\boldsymbol{i}} = \gamma_{\boldsymbol{j}}$ for $\boldsymbol{j} \in D(\boldsymbol{i})$, the following relationship can be derived from \eqref{defn-gs-f}.
\begin{equation} \label{eq-parameter}
    |D(\boldsymbol{i})| = \sum_{\boldsymbol{j} \in D(\boldsymbol{i})} 
    F^{-1} \left( \boldsymbol{u}_{\boldsymbol{j}}^{\top} \boldsymbol{\alpha} 
    + \boldsymbol{u}_{\boldsymbol{j}}^{\top} \boldsymbol{B} \boldsymbol{u}_{\boldsymbol{j}} 
    + \gamma_{\boldsymbol{j}} \right).
\end{equation}
The dependency structures among conditional probabilities of each cell given its symmetric set $D(\boldsymbol{i})$, given by
\begin{equation}
    \pi_{\boldsymbol{i}}^{c} = \frac{\pi_{\boldsymbol{i}}}{\smashoperator[r]{\sum\limits_{\boldsymbol{j} \in D(\boldsymbol{i})}} \pi_{\boldsymbol{j}}}, 
    \quad \boldsymbol{i} \in \mathcal{I}_{V},
\end{equation}
can be expressed using \eqref{eq-parameter} as follows:

\begin{enumerate} [label=(\roman*)] \itemsep=4mm
\item If $f(x) = x\log x - x + 1$, then the $f$-divergence reduces to the KL-divergence, $F^{-1}(x)=e^{x}$, and the GS$[f]$ model \eqref{defn-gs-f} becomes
\begin{equation}
    \pi_{\boldsymbol{i}} = \pi_{\boldsymbol{i}}^{S}
    \exp{\Bigl( \boldsymbol{u}_{\boldsymbol{i}}^{\top} \boldsymbol{\alpha} + \boldsymbol{u}_{\boldsymbol{i}}^{\top} \boldsymbol{B} \boldsymbol{u}_{\boldsymbol{i}} + \gamma_{\boldsymbol{i}} \Bigr)}, 
    \quad \boldsymbol{i} \in \mathcal{I}_{V},
\end{equation}
which, with the help of \eqref{eq-parameter}, is transformed to
\begin{equation} \label{defn-gs}
    \pi_{\boldsymbol{i}}^{c} = \frac{\theta_{\boldsymbol{i}}}{\smashoperator[r]{\sum\limits_{\boldsymbol{j} \in D(\boldsymbol{i})}} \theta_{\boldsymbol{j}}}, 
    \quad \boldsymbol{i} \in \mathcal{I}_{V},
\end{equation}
with 
\begin{equation*}
    \theta_{\boldsymbol{i}} = \exp{\Bigl( \boldsymbol{u}_{\boldsymbol{i}}^{\top} \boldsymbol{\alpha} + \boldsymbol{u}_{\boldsymbol{i}}^{\top} \boldsymbol{B} \boldsymbol{u}_{\boldsymbol{i}} \Bigr)}.
\end{equation*}
The model defined by \eqref{defn-gs} will be referred to as the \textit{Gaussian symmetry} (GS) model. 
Specifically, the GS model represents conditional probabilities with $\{\theta_{\boldsymbol{i}}\}$ in a simple ratio-based form, making the dependency structure both intuitive and transparent.
It can be easily shown that the GS model is equivalent to the GLS model defined as
\begin{equation}
    \pi_{\boldsymbol{i}} = \left( \prod_{s=1}^{T} \alpha_{s}^{i_s} \right) \left( \prod_{s=1}^{T} \beta_{s}^{i_s^2} \right) \left(\prod_{s=1}^{T-1} \prod_{t=s+1}^{T} \gamma_{st}^{i_s i_t} \right) \psi_{\boldsymbol{i}},
    \quad \boldsymbol{i} \in \mathcal{I}_{V},
\end{equation}
where $\psi_{\boldsymbol{i}} = \psi_{\boldsymbol{j}}$ for $\boldsymbol{j} \in D(\boldsymbol{i})$. 
For identifiability, we may set, for example, $\alpha_T = \beta_T = \gamma_{(T-1)T} = 1$. 
We also note that when the scores are integers given by $\{u_i=i\}$, the GS model reduces to the GLS model.

For simplicity and consistency in subsequent discussions, we shall refer to the GLS model as the GS model throughout the paper.
This nomenclature emphasizes its foundational role in describing departures from the S model while aligning with the naming convention of the proposed GS$[f]$ model.
We can easily verify that the GS model yields the ELS model \citep{tomizawa1991Extended} when all off-diagonal elements of $\boldsymbol{B}$ are equal. 
Furthermore, imposing the additional condition that all diagonal elements of $\boldsymbol{B}$ are also equal recovers the LS model \citep{agresti1983Simple}. 
These special cases highlight the flexibility of this framework in capturing various dependency structures. 

Moreover, the following relationship can be derived from equation \eqref{defn-gs}.
\begin{equation} \label{eq-gs-ratio}
    \frac{\pi_{\boldsymbol{i}}^{c}}{\pi_{\boldsymbol{j}}^{c}} = \frac{\pi_{\boldsymbol{i}}}{\pi_{\boldsymbol{j}}}
    = \frac{\theta_{\boldsymbol{i}}}{\theta_{\boldsymbol{j}}},
\end{equation}
where $\boldsymbol{j} \in D(\boldsymbol{i})$ for any $\boldsymbol{i} \in \mathcal{I}_{V}$. 
From equation \eqref{eq-gs-ratio}, the ratio of conditional or cell probabilities at symmetric positions is expressed in terms of the latent parameters $\{\theta_{\boldsymbol{i}}\}$.
Furthermore, by taking the logarithm of both sides, we can verify that the ratio of cell probabilities is represented as a linear relationship in terms of the logarithmic function.

\item If $f(x)=(x-1)^2/2$, then the $f$-divergence reduces to the Pearson's $\chi^2$ divergence, $F^{-1}(x)=x+1$, and the GS$[f]$ model \eqref{defn-gs-f} becomes
\begin{equation*}
    \pi_{\boldsymbol{i}} = \pi_{\boldsymbol{i}}^{S}
    \Bigl(\boldsymbol{u}_{\boldsymbol{i}}^{\top} \boldsymbol{\alpha} + \boldsymbol{u}_{\boldsymbol{i}}^{\top} \boldsymbol{B} \boldsymbol{u}_{\boldsymbol{i}} + \gamma_{\boldsymbol{i}} + 1\Bigr), 
    \quad \boldsymbol{i} \in \mathcal{I}_{V}.
\end{equation*}
Considering \eqref{eq-parameter}, the model can be verified to reduce to
\begin{equation} \label{defn-gs-pearsonian}
    \pi_{\boldsymbol{i}}^{c} = 
    \frac{1}{|D(\boldsymbol{i})|} + \theta_{\boldsymbol{i}}
    - \frac{1}{|D(\boldsymbol{i})|} \smashoperator[r]{\sum_{\boldsymbol{j} \in D(\boldsymbol{i})}} \theta_{\boldsymbol{j}}, 
    \quad \boldsymbol{i} \in \mathcal{I}_{V},
\end{equation}
with
\begin{equation*}
    \theta_{\boldsymbol{i}} = \dfrac{1}{|D(\boldsymbol{i})|} \Bigl( \boldsymbol{u}_{\boldsymbol{i}}^{\top} \boldsymbol{\alpha} + \boldsymbol{u}_{\boldsymbol{i}}^{\top} \boldsymbol{B} \boldsymbol{u}_{\boldsymbol{i}} \Bigr).
\end{equation*}
This model represents a departure from the complete symmetry in the difference-based form.
Under the same constraints as the GS$[f]$ model, this model \eqref{defn-gs-pearsonian} is the closest to the S model when the divergence is measured by Pearson's $\chi^2$ divergence.
Thus, we shall refer to this model as the \textit{Pearsonian} GS (PGS) model, where the conditional probability relationship is given as follows.
\begin{equation} \label{eq-gls-pearsonian}
    \pi_{\boldsymbol{i}}^{c} - \pi_{\boldsymbol{j}}^{c} = \theta_{\boldsymbol{i}} - \theta_{\boldsymbol{j}},
\end{equation}
where $\boldsymbol{j} \in D(\boldsymbol{i})$ for any $\boldsymbol{i} \in \mathcal{I}_{V}$.
Equation \eqref{eq-gls-pearsonian} shows that the PGS model represents the structure of the difference between two conditional probabilities, which is linearly expressed in terms of the latent parameters $\{\theta_{\boldsymbol{i}}\}$.
Contrary to the GS model, which transforms conditional probabilities on a logarithmic scale, the PGS model directly represents the differences in conditional probabilities as a linear relationship.

\item If $f(x)= 2(\sqrt{x}-1)^2$, then the $f$-divergence reduces to the Hellinger distance, $F^{-1}(x)=(1-x/2)^{-2}$, and the GS$[f]$ model \eqref{defn-gs-f} becomes
\begin{equation} \label{defn-gs-hellinger}
    \pi_{\boldsymbol{i}} = \pi_{\boldsymbol{i}}^{S}
    \Biggl(- \frac{1}{2} \Bigl( \boldsymbol{u}_{\boldsymbol{i}}^{\top} \boldsymbol{\alpha}
    + \boldsymbol{u}_{\boldsymbol{i}}^{\top} \boldsymbol{B} \boldsymbol{u}_{\boldsymbol{i}}
    + \gamma_{\boldsymbol{i}} \Bigr)
    + 1 \Biggr)^{-2}, 
    \quad \boldsymbol{i} \in \mathcal{I}_{V}.
\end{equation}
Under the same constraints as the GS$[f]$ model, model \eqref{defn-gs-hellinger} is the closest to the S model when the divergence is measured by Hellinger distance.
Thus, we shall refer to this model as the \textit{Hellinger} GS (HGS) model, where the conditional probability relationship is given as follows.
\begin{equation} \label{eq-gs-hellinger}
    (\pi_{\boldsymbol{i}}^{c})^{-\frac{1}{2}} - (\pi_{\boldsymbol{j}}^{c})^{-\frac{1}{2}} = \theta_{\boldsymbol{i}} - \theta_{\boldsymbol{j}},
\end{equation}
where $\boldsymbol{j} \in D(\boldsymbol{i})$ for any $\boldsymbol{i} \in \mathcal{I}_{V}$ and
\begin{equation}
    \theta_{\boldsymbol{i}} = - \frac{1}{2} \sqrt{|D(\boldsymbol{i})|} \Bigl( \boldsymbol{u}_{\boldsymbol{i}}^{\top} \boldsymbol{\alpha} + \boldsymbol{u}_{\boldsymbol{i}}^{\top} \boldsymbol{B} \boldsymbol{u}_{\boldsymbol{i}} \Bigr).
\end{equation}

\item If the $f$-function is chosen as $f_{\lambda}$ (as denoted in \eqref{eq-power-func}), then the $f$-divergence reduces to the Cressie-Read power divergence, and we obtain
\begin{equation}
    F_{\lambda}^{-1}(x) = 
    \begin{cases}
      (\lambda x + 1)^{\frac{1}{\lambda}},
      & \text{if} \ \lambda \neq 0, \\
      e^{x},  
      & \text{if} \ \lambda = 0,
    \end{cases}
\end{equation}
where $\lambda$ is a real-valued parameter. 
If $\lambda=0$, it corresponds to the GS model, whereas if $\lambda \neq 0$, the GS$[f]$ model \eqref{defn-gs-f} becomes
\begin{equation} \label{defn-gs-power}
    \pi_{\boldsymbol{i}} = \pi_{\boldsymbol{i}}^{S} 
    \Biggl( \lambda \Bigl( \boldsymbol{u}_{\boldsymbol{i}}^{\top} \boldsymbol{\alpha} + \boldsymbol{u}_{\boldsymbol{i}}^{\top} \boldsymbol{B} \boldsymbol{u}_{\boldsymbol{i}} + \gamma_{\boldsymbol{i}} \Bigr)
    + 1 \Biggr)^{1/\lambda}, 
    \quad \boldsymbol{i} \in \mathcal{I}_{V}.
\end{equation}
Under the same constraints as the GS$[f]$ model, model \eqref{defn-gs-power} is the closest to the S model when the divergence is measured by the Cressie-Read power divergence.
This model also reduces to the PGS model when $\lambda=1$ and to the HGS model when $\lambda = -1/2$. 

In a similar manner to the above, the following relationship can be derived.
\begin{equation} \label{eq-gls-power}
    (\pi_{\boldsymbol{i}}^{c})^{\lambda} - (\pi_{\boldsymbol{j}}^{c})^{\lambda}  = \theta_{\boldsymbol{i}} - \theta_{\boldsymbol{j}},
    \quad \lambda \neq 0,
\end{equation}
where $\boldsymbol{j} \in D(\boldsymbol{i})$ for any $\boldsymbol{i} \in \mathcal{I}_{V}$ and
\begin{equation}
    \theta_{\boldsymbol{i}} = \frac{\lambda}{|D(\boldsymbol{i})|^{\lambda}} 
    \Bigl( \boldsymbol{u}_{\boldsymbol{i}}^{\top} \boldsymbol{\alpha} + \boldsymbol{u}_{\boldsymbol{i}}^{\top} \boldsymbol{B} \boldsymbol{u}_{\boldsymbol{i}} \Bigr).
\end{equation}
\end{enumerate}

\noindent Note that the models in ($\mathrm{i}$)-($\mathrm{iv}$) require additional parameter constraints, such as $\gamma_{\boldsymbol{i}} = \gamma_{\boldsymbol{j}}$ for $\boldsymbol{j} \in D(\boldsymbol{i})$ and $\alpha_{T} = \beta_{TT} = \beta_{(T-1)T} = 0$, as noted in Definition \ref{defn-1} to ensure identifiability.

We can interpret the conditional structure of the GS$[f]$ model \eqref{defn-gs-f} analogously to the Ising model framework \citep{ising1925Beitrag}. 
Although the classical Ising model focuses on the interactions between adjacent lattice points, the GS$[f]$ model is designed to describe the interactions among cells located at symmetric positions in a contingency table. 
In this setting, the parameters $\{ \theta_{\boldsymbol{i}} \}$ serve as latent representations of the influence of each cell relative to its symmetric counterparts. 
We refer to these as the \textit{potential} parameters, reflecting their roles in capturing the relative contributions of each cell. 
Hence, the GS$[f]$ model could be viewed as an energy-based model for describing the dependency structures in categorical data.

\subsection{Parameter Interpretation and Role of $f$-divergence} \label{subsec:interpretation}
In this section, we provide an interpretation of parameters $\boldsymbol{\alpha}$ and $\boldsymbol{B}$ in the proposed model and discuss the role of the divergence measure (i.e., the choice of the $f$-function) in relating the model to a log-linear framework.
For simplicity, we first consider a saturated log-linear model for a three‐way contingency table, where $\boldsymbol{i}=(i_1,i_2,i_3)$ and $V=\{1,2,3\}$:
\begin{equation} \label{eq-loglinear}
    \pi_{\boldsymbol{i}} = \dfrac{\Lambda_{\boldsymbol{i}}}{\sum\limits_{\boldsymbol{i} \in \mathcal{I}_{V}} \Lambda_{\boldsymbol{i}}},
    \quad \boldsymbol{i} \in \mathcal{I}_{V},
\end{equation}
with
\begin{equation} \label{eq-loglinear-param}
    \Lambda_{\boldsymbol{i}} 
    = \exp \Biggl( 
    \underbrace{\sum_{s \in V} \lambda_{s}(i_s)}_{\text{Main effects}}
    + \underbrace{\sum_{s<t} \lambda_{st}(i_s, i_t) + \lambda_{123} (\boldsymbol{i})}_{\text{Interactions}}
    \Biggr).
\end{equation}
Under typical identifiability constraints, for each $s \in V$ we impose
\begin{gather*}
    \sum_{i_{s} \in \mathcal{I}} \lambda_{s} (i_{s}) = 0, \quad 
    \sum_{i_{s} \in \mathcal{I}} \lambda_{123} (\boldsymbol{i}) = 0, \quad s \in V, \\
    \sum_{i_{s} \in \mathcal{I}} \lambda_{st} (i_{s}, i_{t}) 
    = \sum_{i_{t} \in \mathcal{I}} \lambda_{st} (i_{s}, i_{t}) = 0, 
    \quad s,t \in V, \ s<t.
\end{gather*}
Here, $\lambda_{s}(i_s)$ represents the main effects and terms $\lambda_{st}(i_s,i_t)$ and $\lambda_{123}(\boldsymbol{i})$ capture higher‐order interactions \citep[see][]{agresti2012Categorical}.

In the formulation of GS model \eqref{defn-gs}, score vectors $\boldsymbol{u}_{\boldsymbol{i}}$ are fixed with normalization conditions, such that the free parameters are $\boldsymbol{\alpha}$ and $\boldsymbol{B}$. 
Note that any intercept term cancels out in the ratio defining $\pi_{\boldsymbol{i}}^{c}$, which is analogous to the cancellation of $\lambda$ in \eqref{eq-loglinear-param}.
By comparing this with the log-linear parameterization in \eqref{eq-loglinear-param}, we can interpret the two components as follows:
Term $\boldsymbol{u}_{\boldsymbol{i}}^{\top}\boldsymbol{\alpha}$ provides a linear approximation of the main effects, whereas the quadratic term $\boldsymbol{u}_{\boldsymbol{i}}^{\top}\boldsymbol{B}\,\boldsymbol{u}_{\boldsymbol{i}}$ can be separated into two components. 
Its diagonal part further refines the approximation of the main effects by capturing nonlinear contributions, whereas its off-diagonal part approximates the first-order interactions corresponding to the interaction terms.
It is important to emphasize that the GS model is constructed within a log-linear framework but conditioned on a symmetric structure. 
That is, rather than modeling the full joint probability, the GS model characterizes the conditional probabilities within each symmetric cell set $D(\boldsymbol{i})$. 

Furthermore, when an arbitrary $f$-function is employed, the proposed model can be written in a more general form as follows:
\begin{equation} \label{eq-gls-f-param}
    F\Bigl(|D(\boldsymbol{i})|\pi_{\boldsymbol{i}}^{c}\Bigr) 
    = \boldsymbol{u}_{\boldsymbol{i}}^{\top}\boldsymbol{\alpha} 
    + \boldsymbol{u}_{\boldsymbol{i}}^{\top}\boldsymbol{B}\,\boldsymbol{u}_{\boldsymbol{i}}
    + \gamma_{\boldsymbol{i}}, \quad \boldsymbol{i}\in \mathcal{I}_{V},
\end{equation}
where the link function $F$ is determined by the chosen $f$-divergence. 
Therefore, the $f$-function is analogous to a link function in generalized linear models, allowing the framework to extend beyond the conventional log-linear approach.
Notably, the interpretation of the parameters remains consistent with the case of the KL-divergence. 
Specifically, the term $\boldsymbol{u}_{\boldsymbol{i}}^{\top}\boldsymbol{\alpha}$ continues to approximate the main effects linearly, whereas the quadratic term $\boldsymbol{u}_{\boldsymbol{i}}^{\top}\boldsymbol{B} \boldsymbol{u}_{\boldsymbol{i}}$ captures nonlinear effects and first-order interactions. 
This consistency arises because the underlying structural representation of the conditional probabilities is preserved regardless of the choice of the $f$-function. 
Therefore, the proposed model maintains a unified interpretation framework across different $f$-divergences.

\section{Overview of Complete Symmetry} 
\label{sec:properties}
In this section, we explore the relationship between the S model and the proposed GS$[f]$ model, focusing on their structural and statistical properties.
Section \ref{subsec:decomposition} provides the necessary and sufficient conditions for the S model and elucidates its decomposition into less-restrictive component models. 
Section \ref{subsec:test} discusses the properties of the test statistics associated with the S model, focusing on the partitioning of goodness-of-fit statistics and their implications for model assessment. 

For the empirical analysis in Section \ref{sec:example}, we define the ELS$[f]$ model (where all off-diagonal elements of $\boldsymbol{B}$ are equal) and LS$[f]$ model (where all off-diagonal and diagonal elements of $\boldsymbol{B}$ are equal) as extensions of the ELS and LS models \citep{tomizawa1991Extended, agresti1983Simple}, respectively.
Although similar properties can be demonstrated for these nested models, we omit these discussions for brevity.

\subsection{Component Models in Complete Symmetry}
\label{subsec:decomposition}
In this section, we introduce symmetric models that complement the asymmetric GS$[f]$ model and provide a framework for reconstructing complete symmetry.
This reconstruction is crucial for partitioning the test statistics discussed in Section \ref{subsec:test} and enhancing our understanding of the structure of multivariate categorical data. 
To achieve this, we first introduced fundamental symmetric models.

First, because the proposed GS$[f]$ model allows flexibility in the marginal means and the covariance matrix regarding the score transformation ($g(i)=u_i$), we introduce the \textit{second-moment equality} (ME$_2$) model, which imposes symmetric structures on these quantities.  
The ME$_2$ model is defined as follows:
\begin{equation}
    \mu_1 = \mu_s, \ \sigma_{1}^{2} = \sigma_{s}^{2} \quad \text{and} \quad \rho_{12} = \rho_{st}, \quad s,t \in V,
\end{equation}
where
\begin{equation}
   \mu_{u}   = \mathrm{E}\!\bigl[g(X_{u})\bigr],\quad
   \sigma_{u}^{2} = \mathrm{Var}\!\bigl[g(X_{u})\bigr],\quad
   \rho_{st} = \frac{\mathrm{E}\!\bigl[g(X_{s})\,g(X_{t})\bigr] - \mu_{s}\mu_{t}}{\sigma_{s}\sigma_{t}}.
\end{equation}
For $T \ge 3$, the following results describe the construction of the S model.

\begin{theorem} \label{thm-2}
    The S model holds if and only if both the GS$[f]$ and ME$_{2}$ models hold.
\end{theorem}

The proof of this theorem is provided in Appendix \ref{appendix-proofs}.
This theorem plays a crucial role in the analysis of test statistics in Section \ref{subsec:test}.
In addition to the ME$_{2}$ model, introducing marginal models that focus on specific aspects of complete symmetry is beneficial.
First, consider the \textit{marginal mean equality} (ME) model, defined as 
\begin{equation*}
    \mu_{1} = \mu_{u}, \quad u \in V.
\end{equation*}
Second, consider the \textit{marginal variance equality} (VE) model, defined as
\begin{equation*}
    \sigma_{1}^{2} = \sigma_{u}^{2}, \quad u \in V.
\end{equation*}
Finally, we consider the \textit{marginal correlation equality} (CE) model, defined as
\begin{equation*}
    \rho_{12} = \rho_{st}, \quad s,t \in V.
\end{equation*}
These models address different aspects of complete symmetry. 
Together with the GS$[f]$ model, they form a comprehensive framework for decomposing the S model.
The following result highlights the decomposition of the S model into the following groups of models:

\begin{corollary} \label{cor-2}
    The S model holds if and only if all GS$[f]$, ME, VE, and CE models hold.
\end{corollary}

This corollary is particularly useful for identifying the sources of the lack-of-fit in the S model. 
These results hold regardless of the specific combination of models used as long as the conditions of the ME$_{2}$ model are satisfied.

\begin{remark}
The validity of this decomposition can be explained as follows.
The proposed model was constructed to minimize $f$-divergence under the given marginal means and covariance structures.
By parameterizing these structures, the model adopts a framework similar to that of a multivariate normal distribution characterized by mean parameters and a covariance matrix.
Restoring symmetry in a complementary manner enables the reconstruction of complete symmetry.
\end{remark}

\subsection{Properties of Test Statistics under Complete Symmetry}
\label{subsec:test}

\subsubsection{Goodness-of-fit Statistics}
\label{subsubsec:goodness-of-fit}
Let $n_{\boldsymbol{i}}$ denote the observed frequency in the $\boldsymbol{i}$th cell for $\boldsymbol{i} \in \mathcal{I}_{V}$ in the $r^T$ contingency table.
We assume that the cell counts follow a multinomial distribution.
Let $m_{\boldsymbol{i}}$ and $\hat m_{\boldsymbol{i}}$ denote the expected frequency in the $\boldsymbol{i}$th cell and its maximum likelihood estimate (MLE) under a given model. 
The MLEs of the expected frequencies can be obtained using the Newton-Raphson method to solve the log-likelihood equations.
For instance, we can use packages \texttt{Rsolnp} and \texttt{mph.fit} \citep{lang2005Homogeneous} in \textsf{R} to calculate the MLEs under any model, as demonstrated in Section \ref{sec:example}.

Various test statistics can be used to assess the goodness-of-fit of a model, such as the likelihood ratio chi-square statistic, power-divergence statistics \citep{cressie1984Multinomial}, and $\phi$-divergence test statistics \citep{pardo2010Minimum, pardo2010Approach}.
The $\phi$-divergence test statistics provide a unified framework for assessing goodness-of-fit by generalizing several divergence measures, including the likelihood ratio and Pearson's chi-square statistics, as special cases.
For example, the likelihood ratio chi-square statistic is defined as 
\begin{equation}
    G^{2}(\text{M}) = 2 \sum_{\boldsymbol{i} \in \mathcal{I}_{V}} n_{\boldsymbol{i}} \log \left( \frac{n_{\boldsymbol{i}}}{\hat m_{\boldsymbol{i}}} \right),
\end{equation}
where, under model $\text{M}$, $G^{2}(\text{M})$ asymptotically follows a $\chi^{2}_{\textit{df}}$ distribution, where $\textit{df}$ denotes the corresponding degrees of freedom.
Table \ref{table:data_type} presents the \textit{df} values for the various models in the context of the $r^T$ contingency table.
The properties of the test statistics for several models have been extensively studied \citep{darroch1963Testing, tomizawa2007Analysis}.
If a model $\mathrm{M}_3$ can be decomposed into two models $\mathrm{M}_1$ and $\mathrm{M}_2$, and the likelihood ratio chi-square statistic satisfies the additive relation
\begin{equation} \label{eq-g2}
    G^2(\textup{M}_3) = G^2(\textup{M}_1) + G^2(\textup{M}_2) + o_p(1),
\end{equation}
where \textit{df} for M$_3$ is equal to the sum of the \textit{df} for M$_1$ and M$_2$, then, when both $\mathrm{M}_1$ and $\mathrm{M}_2$ are accepted with high probability, $\mathrm{M}_3$ would also be accepted \citep[see][]{aitchison1962Largesample}.

\begin{table}[ht]
  \caption{Degrees of freedom for various models in the $r^T$ contingency table}
  \label{table:data_type}
  \centering
  \begin{tabular}{ll}
    \toprule
    Models & \textit{df} \\
    \midrule
    S        & $r^T - L^{(r,T)}$ \\
    GS$[f]$  & $r^T - L^{(r,T)} - (T^2+3T-6)/2$ \\
    ELS$[f]$ & $r^T - L^{(r,T)} - 2T + 2$ \\
    LS$[f]$  & $r^T - L^{(r,T)} - T + 1$ \\
    ME$_{2}$ & $(T^2+3T-6)/2$ \\
    ME       & $T-1$ \\
    VE       & $T-1$ \\
    CE       & $(T^2-T-2)/2$ \\
    \bottomrule \addlinespace[3pt]
    \multicolumn{2}{c}{\textit{Note:} $L^{(r,T)}=\displaystyle{\binom{r+T-1}{T} = \frac{(r+T-1)!}{(r-1)!T!}}$}
  \end{tabular}
\end{table}

\subsubsection{Preparations}
\label{subsubsec:preparations}
To facilitate the discussion of the properties of goodness-of-fit statistics, we introduce some additional notation. 
Let $\mathcal{I}_{V}^{S}$ be a subset of $\mathcal{I}_{V}$ consisting of indices in non-decreasing order, which correspond to parameters $\{\gamma_{\boldsymbol{i}}\}$:
\begin{equation}
    \mathcal{I}_{V}^{S} = \{(i_1,\dots, i_T) \ | \ i_1 \le \cdots \le i_T\} \subset \mathcal{I}_{V}.
\end{equation}
For a set $V=\{1,\dots,T\}$, we define a set $\mathcal{V}_{2}$ of ordered pairs corresponding to the off-diagonal parameters of a symmetric matrix $\boldsymbol{B}$ as
\begin{equation}
    \mathcal{V}_{2} = \{(v_1, v_2) \ | \ v_1 < v_2 \ \land \ v_1 < T-1 \} \subset V \times V.
\end{equation}

We also introduce the concept of \textit{lexicographical ordering} ($\preceq$), defined as
\begin{equation*}
    \forall \boldsymbol{v},\boldsymbol{v}' \in \mathbb{R}^{n}, \ \boldsymbol{v} \preceq \boldsymbol{v}' \ \Leftrightarrow \ \exists i \in \{1,\dots,n\}, \ (\forall j<i, v_j = v_{j}') \land (v_{i} \le v_{i}').
\end{equation*}
This ordering allows us to define $(\mathcal{I}_{V}, \preceq)$, $(\mathcal{I}_{V}^{S}, \preceq)$ and $(\mathcal{V}_{2}, \preceq)$ as totally ordered sets.
For an arbitrary totally ordered set $\mathcal{A}$ with lexicographical ordering $\preceq$, the elements can be enumerated as a sequence $(\boldsymbol{a}_{n})_{n \in \mathbb{N}}$ where $\boldsymbol{a}_{1} \preceq \dots \preceq \boldsymbol{a}_{|\mathcal{A}|}$.
Similarly, we construct sequences of elements in ascending order for each defined set as follows:
\begin{itemize} \itemsep=2mm
    \item For $\mathcal{I}_{V}$: \ $(\boldsymbol{\iota}_{1}, \dots, \boldsymbol{\iota}_{|\mathcal{I}_{V}|})$
    where $\boldsymbol{\iota}_{1} \preceq \cdots \preceq \boldsymbol{\iota}_{|\mathcal{I}_{V}|}$

    \item For $\mathcal{I}_{V}^{S}$: \ $(\boldsymbol{\sigma}_{1}, \dots, \boldsymbol{\sigma}_{|\mathcal{I}_{V}^{S}|})$ 
    where $\boldsymbol{\sigma}_{1} \preceq \cdots \preceq \boldsymbol{\sigma}_{|\mathcal{I}_{V}^{S}|}$

    \item For $\mathcal{V}_{2}$: \ $(\boldsymbol{\nu}_{1}, \dots, \boldsymbol{\nu}_{|\mathcal{V}_{2}|})$
    where $\boldsymbol{\nu}_{1} \preceq \cdots \preceq \boldsymbol{\nu}_{|\mathcal{V}_{2}|}$
\end{itemize}

After lexicographically sorting cells $\boldsymbol{\iota}_1 \preceq \dots \preceq \boldsymbol{\iota}_{|\mathcal{I}_{V}|}$, we define the $h$th score vector for a table with $r$ ordinal categories as
\begin{equation}
    \boldsymbol{s}_{h} = \boldsymbol{1}_{r^{h-1}} \otimes (u_{1}, \dots, u_{r})^\top \otimes \boldsymbol{1}_{r^{T-h}}, 
    \quad h \in V,
\end{equation}
where $\boldsymbol{1}_{s}$ is the $s \times 1$ vector of ones and ``$\otimes$'' denotes the Kronecker product \citep{magnus2019Matrix}.
Each resulting vector $\boldsymbol{s}_{h} \in \mathbb{R}^{|\mathcal{I}_{V}|}$ assigns appropriate score values to the $h$th variable $X_h$ across all possible cell combinations.

Based on these score vectors, we introduce three types of difference score vectors essential for our analysis.
The first- and second-order differences are defined for each $h \in \{1, \dots, T-1\}$ as:
\begin{equation}
  \boldsymbol{\Delta}_{1}^{(h)} =
  \boldsymbol s_{h} - \boldsymbol s_{h+1}, 
  \quad
  \boldsymbol{\Delta}_{2}^{(h)} =
  (\boldsymbol{s}_{h} \odot \boldsymbol{s}_{h})
  - (\boldsymbol{s}_{h+1} \odot \boldsymbol{s}_{h+1}),
\end{equation}
where ``$\odot$'' denotes the Hadamard product.
Next, for each pair $\boldsymbol{\nu}_{k} = (s, t) \in \mathcal{V}_2$, we define the second-order interaction differences as
\begin{equation}
  \boldsymbol{\Delta}_{\mathcal{V}_2}^{(k)}
  = (\boldsymbol{s}_{s} \odot \boldsymbol{s}_{t})
  - (\boldsymbol{s}_{s^{+}} \odot \boldsymbol{s}_{t^{+}}),
  \quad k = 1, \dots, |\mathcal{V}_2|,
\end{equation}
where $\boldsymbol{\nu}_{k+1} = (s^{+}, t^{+})$ represents the next pair in lexicographic order of $\mathcal{V}_2$. 
Note that we introduce $\boldsymbol{\nu}_{|\mathcal{V}_{2}|} = (T-1, T)$ for convenience, although this pair is not included in $\mathcal{V}_{2}$.

\subsubsection{Partition of Test Statistics}
\label{subsubsec:partition}
Let
\begin{align*}
    \boldsymbol{\pi} &= \left(\pi_{\boldsymbol{\iota}_1}, \dots, \pi_{\boldsymbol{\iota}_{|\mathcal{I}_{V}|}} \right)^\top, \\
    \boldsymbol{\pi}^{S} &= \left(\pi_{\boldsymbol{\iota}_1}^{S}, \dots, \pi_{\boldsymbol{\iota}_{|\mathcal{I}_{V}|}}^{S} \right)^\top, \\
    F \biggl( \frac{\boldsymbol{\pi}}{\boldsymbol{\pi}^S} \biggr) &= \left( F \Biggl( \frac{\pi_{\boldsymbol{\iota}_{1}}}{\pi_{\boldsymbol{\iota}_{1}}^{S}} \Biggr),\dots, F \Biggl( \frac{\pi_{\boldsymbol{\iota}_{|\mathcal{I}_{V}|}}}{\pi_{\boldsymbol{\iota}_{|\mathcal{I}_{V}|}}^{S}} \Biggr) \right)^\top, \\
    \boldsymbol{\vartheta} &= (\boldsymbol{\alpha}^{\prime}, \boldsymbol{\beta}_{\text{diag}}^{\prime}, \boldsymbol{\beta}_{\text{off-diag}}^{\prime}, \boldsymbol{\gamma}^{\prime})^\top,
\end{align*}
where the parameter vectors are defined as:
\begin{align*}
    \boldsymbol{\alpha}^{\prime} &= (\alpha_{1}^{\prime},\dots, \alpha_{T-1}^{\prime}), \\ 
    \boldsymbol{\beta}_{\text{diag}}^{\prime} &= \left(\beta_{11}^{\prime},\dots, \beta_{(T-1)(T-1)}^{\prime} \right), \\
    \boldsymbol{\beta}_{\text{off-diag}}^{\prime} &= \left(\beta_{\boldsymbol{\nu}_{1}}^{\prime}, \dots, \beta_{\boldsymbol{\nu}_{|\mathcal{V}_{2}|}}^{\prime} \right), \\
    \boldsymbol{\gamma}^{\prime} &= \left( \gamma_{\boldsymbol{\sigma}_{1}}^{\prime}, \dots, \gamma_{\boldsymbol{\sigma}_{|\mathcal{I}_{V}^{S}|}}^{\prime} \right).
\end{align*}

Using the previously defined difference vectors, we construct the following design matrix:
\begin{equation}
  \boldsymbol{X} =
  \biggl(
    \underbrace{\boldsymbol{\Delta}_{1}^{(1)}, \dots, \boldsymbol{\Delta}_{1}^{(T-1)} }_{\text{$\boldsymbol{\alpha}^{\prime}$-block}},
    \
    \underbrace{\boldsymbol{\Delta}_{2}^{(1)}, \dots, \boldsymbol{\Delta}_{2}^{(T-1)} }_{\text{$\boldsymbol{\beta}_{\text{diag}}^{\prime}$-block}},
    \
    \underbrace{\boldsymbol{\Delta}_{\mathcal{V}_2}^{(1)}, \dots, \boldsymbol{\Delta}_{\mathcal{V}_2}^{(|\mathcal{V}_{2}|)} }_{\text{$\boldsymbol{\beta}_{\text{off-diag}}^{\prime}$-block}},
    \
    \underbrace{\boldsymbol X^{S}}_{\text{$\boldsymbol{\gamma}^{\prime}$-block}}
  \biggr),
\end{equation}
where $\boldsymbol{X}^{S}$ is the $r^T \times L^{(r,T)}$ matrix of $1$ or $0$ elements that construct the complete symmetry structure.
This matrix satisfies $\boldsymbol{X}^{S} \boldsymbol{1}_{L^{(r,T)}} = \boldsymbol{1}_{r^T}$, which ensures that each cell belongs to exactly one symmetric class.
The design matrix $\boldsymbol{X}$ has dimension $r^T \times L$, where $L = d_2 + L^{(r,T)}$ and $d_2 = (T^2+3T-6)/2$ represents the total number of parameters for $\boldsymbol{\alpha}^{\prime}$, $\boldsymbol{\beta}_{\text{diag}}^{\prime}$, and $\boldsymbol{\beta}_{\text{off-diag}}^{\prime}$.
Matrix $\boldsymbol{X}$ has the full column rank.

With this construction, the GS$[f]$ model can be expressed in a linear form as
\begin{equation} \label{defn-gs-f-linear}
F \left( \frac{\boldsymbol{\pi}}{\boldsymbol{\pi}^{S}} \right)
= \boldsymbol{X} \boldsymbol{\vartheta}.
\end{equation}
This linear formulation provides an equivalent characterization of the GS$[f]$ model in terms of the constrained optimization, as shown in Appendix \ref{appendix-proofs}. 
It serves as the foundation for the statistical analysis that follows.

To investigate the structure of the goodness-of-fit statistics, we denote the linear space spanned by the columns of matrix $\boldsymbol{X}$ by $S(\boldsymbol{X})$ with dimension $L$.
$S(\boldsymbol{X})$ is subspace of $\mathbb{R}^{r^T}$. 
Let $\boldsymbol{U}$ be an $r^T \times d_1$, where $d_1 = r^T - L$, and a full column rank matrix such that the linear space is spanned by the columns of $\boldsymbol{U}$. 
That is, $S(\boldsymbol{U})$ is an orthogonal complement of space $S(\boldsymbol{X})$. Thus, $\boldsymbol{U}^\top \boldsymbol{X} = \boldsymbol{O}_{d_{1},L}$ where $\boldsymbol{O}_{d_{1},L}$ is a $d_{1} \times L$ zero matrix.

Let $\boldsymbol{h}_{1}(\boldsymbol{\pi})$ and $\boldsymbol{h}_{2}(\boldsymbol{\pi})$ be a vector of functions defined by $\boldsymbol{h}_{1}(\boldsymbol{\pi}) = \boldsymbol{U}^\top F ( \boldsymbol{\pi}/\boldsymbol{\pi}^S )$ and $\boldsymbol{h}_{2} (\boldsymbol{\pi}) = \boldsymbol{M} \boldsymbol{\pi}$ with the $d_2 \times r^T$ matrix $\boldsymbol{M} = (\boldsymbol{\Delta}_{1}^{(1)}, \dots, \boldsymbol{\Delta}_{1}^{(T-1)}, \boldsymbol{\Delta}_{2}^{(1)}, \dots, \boldsymbol{\Delta}_{2}^{(T-1)},$     $\boldsymbol{\Delta}_{\mathcal{V}_2}^{(1)}, \dots, \boldsymbol{\Delta}_{\mathcal{V}_2}^{(|\mathcal{V}_{2}|)})^\top$, respectively. 
Note that $\boldsymbol{M}^\top$ belongs to space $S(\boldsymbol{X})$; namely, $S(\boldsymbol{M}^\top) \subset S(\boldsymbol{X})$. From Theorem \ref{thm-2}, the S model is equivalent to the hypothesis $\boldsymbol{h}_{3} (\boldsymbol{\pi}) = (\boldsymbol{h}_{1}^{\top} (\boldsymbol{\pi}),\boldsymbol{h}_{2}^{\top} (\boldsymbol{\pi}))^\top = \boldsymbol{0}_{d_3}$, where $\boldsymbol{0}_{s}$ is the $s \times 1$ zero vector and $d_3 = d_1+d_2$, because the GS$[f]$ model is equivalent to the hypothesis $\boldsymbol{h}_{1} (\boldsymbol{\pi}) = \boldsymbol{0}_{d_1}$ and the ME$_{2}$ model is equivalent to the hypothesis $\boldsymbol{h}_{2} (\boldsymbol{\pi}) = \boldsymbol{0}_{d_2}$.

Let $\boldsymbol{H}_{s} ( \boldsymbol{\pi} )$ for $s=1,2,3$ denote the $d_s \times r^T$ matrix of partial derivatives of $\boldsymbol{h}_{s} (\boldsymbol{\pi})$ with respect to $\boldsymbol{\pi}$, i.e., $\boldsymbol{H}_{s} ( \boldsymbol{\pi} ) = \partial \boldsymbol{h}_{s} ( \boldsymbol{\pi}) / \partial \boldsymbol{\pi}^\top$ \citep[see][]{magnus2019Matrix}. 
Also, let $\boldsymbol{\Sigma}(\boldsymbol{\pi}) = \text{diag}(\boldsymbol{\pi}) - \boldsymbol{\pi} \boldsymbol{\pi}^\top$, where $\rm diag (\boldsymbol{\pi})$ denotes a diagonal matrix with the $i$th component of $\boldsymbol{\pi}$ as the $i$th diagonal component, and $\boldsymbol{p}$ denotes $\boldsymbol{\pi}$ with $\pi_{\boldsymbol{\iota}_{i}}$ replaced by $p_{\boldsymbol{\iota}_{i}}$, where $p_{\boldsymbol{\iota}_{i}} = n_{\boldsymbol{\iota}_{i}} / n$ with $n=\sum_{\boldsymbol{i} \in \mathcal{I}_{V}}  n_{\boldsymbol{i}}$.
Since $\sqrt{n}(\boldsymbol{p} - \boldsymbol{\pi})$ has an asymptotically multivariate normal distribution with mean $\boldsymbol{0}_{r^T}$ and covariance matrix $\boldsymbol{\Sigma}(\boldsymbol{\pi})$, using the delta method, $\sqrt{n}(\boldsymbol{h}_{3} (\boldsymbol{p}) - \boldsymbol{h}_{3} (\boldsymbol{\pi}))$ has asymptotically a multivariate normal distribution with mean $\boldsymbol{0}_{r^T}$ and covariance matrix
\begin{equation*}
    \boldsymbol{H}_{3} (\boldsymbol{\pi}) \boldsymbol{\Sigma}(\boldsymbol{\pi}) \boldsymbol{H}_{3}^\top (\boldsymbol{\pi}) = 
    \begin{bmatrix}
    \boldsymbol{H}_{1} (\boldsymbol{\pi}) \boldsymbol{\Sigma}(\boldsymbol{\pi}) \boldsymbol{H}_{1}^\top (\boldsymbol{\pi}) & \boldsymbol{H}_{1} (\boldsymbol{\pi}) \boldsymbol{\Sigma}(\boldsymbol{\pi}) \boldsymbol{H}_{2}^\top (\boldsymbol{\pi}) \\
    \boldsymbol{H}_{2} (\boldsymbol{\pi}) \boldsymbol{\Sigma}(\boldsymbol{\pi}) \boldsymbol{H}_{1}^\top (\boldsymbol{\pi}) & \boldsymbol{H}_{2} (\boldsymbol{\pi}) \boldsymbol{\Sigma}(\boldsymbol{\pi}) \boldsymbol{H}_{2}^\top (\boldsymbol{\pi})
    \end{bmatrix}.
\end{equation*}
Then,
\begin{equation*}
    \boldsymbol{H}_{1} (\boldsymbol{\pi}) = \boldsymbol{U}^\top \frac{\partial}{\partial \boldsymbol{\pi}^\top} F \left( \frac{\boldsymbol{\pi}}{\boldsymbol{\pi}^S} \right),
\end{equation*}
and
\begin{equation*}
    \boldsymbol{H}_{2} (\boldsymbol{\pi}) = \boldsymbol{M}.
\end{equation*}

Thus, under the hypothesis $\boldsymbol{h}_{3} (\boldsymbol{\pi}) = \boldsymbol{0}_{d_3}$, we obtain
\begin{equation}
    W_3 = W_1 + W_2,
\end{equation}
where
\begin{equation*}
    W_s = n \boldsymbol{h}_{s}^\top (\boldsymbol{p}) \left( \boldsymbol{H}_{s} (\boldsymbol{p}) \boldsymbol{\Sigma}(\boldsymbol{p}) \boldsymbol{H}_{s}^\top (\boldsymbol{p}) \right)^{-1} \boldsymbol{h}_{s}(\boldsymbol{p}).
\end{equation*}
This result follows from the fact that 
\begin{equation} \label{eq-nondiagonal}
    \boldsymbol{H}_{1} (\boldsymbol{\pi}) \boldsymbol{\Sigma}(\boldsymbol{\pi}) \boldsymbol{H}_{2}^\top (\boldsymbol{\pi}) = \boldsymbol{O}_{d_{1},d_{2}},
\end{equation}
in Appendix \ref{appendix-proofs}.
We show that (\rm{i}) $W_1$ is Wald statistic for the GS$[f]$ model, (\rm{ii}) $W_2$ is that for the ME$_{2}$ model, and (\rm{iii}) $W_3$ is that for the S model.
Furthermore, the Wald statistic $W_s$ contains $d_s$ \textit{df} for $s=1,2,3$.
Based on these results, we propose the following theorem for $T \ge 3$.

\begin{theorem} \label{thm-3}
    When the S model holds, the following equivalence holds:
    \begin{equation}
        W(\textup{S}) = W(\textup{GS}[f]) + W(\textup{ME}_{2}),
    \end{equation}
    where $W(\textup{M})$ is the Wald statistic for model M.
\end{theorem}

Given the asymptotic equivalence of the Wald and likelihood ratio statistics \cite[see][]{rao1973Large}, we can extend the result of Theorem \ref{thm-3} to likelihood ratio statistics, leading to the following theorem:

\begin{theorem} \label{thm-4}
    When the S model holds, the following asymptotic equivalence holds:
    \begin{equation*}
        G^{2}(\textup{S}) = G^{2}(\textup{GS}[f]) + G^{2}(\textup{ME}_{2}) + o_p(1).
    \end{equation*}
\end{theorem}

Theorem \ref{thm-4} extends the partitioning result of Theorem \ref{thm-3} to the likelihood ratio statistics that are commonly used.
This theorem suggests that when the S model holds, the likelihood ratio statistic for testing the S model can be asymptotically partitioned into the sum of the likelihood ratio statistics for testing the GS$[f]$ and ME$_{2}$ models. 
This partition allows researchers to identify sources of asymmetry in the data by examining the goodness-of-fit statistics of the individual components. 
It also provides a convenient method for testing the S model by combining the tests for the GS$[f]$ and ME$_{2}$ models.

\section{Simulation Study} 
\label{sec:simulation}
For two-way contingency tables, \cite{agresti1983Simple} indicates that the LS model is adequate for tables underlying bivariate normal distributions with different means and equal variances. 
Subsequently, \cite{tomizawa1991Extended} proposed an ELS model to accommodate unequal marginal variances. 
For multi-way contingency tables, \cite{yamamoto2007Decomposition} introduced the GS model and examined its performance using $4 \times 4 \times 4$ contingency tables generated from underlying trivariate normal distributions with various means, variances, and correlation coefficients.
However, their analysis was limited to a single table for each parameter configuration rather than assessing the empirical power through repeated simulations.

The primary objectives of this simulation study are as follows:
First, we aim to evaluate the empirical power of the test for the proposed model in detecting asymmetric structures underlying multivariate normal distributions, which extend beyond one-shot analysis \citep{yamamoto2007Decomposition}. 
Second, we demonstrate that the GS[$f$] model provides a good fit to the contingency tables generated from multivariate normal distributions across various parameter settings.
Third, we examine whether varying the $f$-function affects the performance of the model.
This indicates flexibility in choosing appropriate divergence measures based on the conditional probability relationships of interest without compromising fit quality.

In this study, we investigated the properties of the following models:
\begin{itemize}
    \item KL-divergence-based models: GS, ELS, and LS
    \item Pearson's $\chi^2$ divergence-based models: PGS, PELS, and PLS
    \item Hellinger distance-based models: HGS, HELS, and HLS
\end{itemize}
Note that model names prefixed with `P' refer to Pearson’s $\chi^2$-divergence-based models, whereas those prefixed with `H' are Hellinger distance-based models.

In this simulation study, we assume random sampling from an underlying trivariate normal distribution characterized by mean vector $(\mu_{1}, \mu_{2}, \mu_{3})$, variance vector $(\sigma_{1}^{2}, \sigma_{2}^{2}, \sigma_{3}^{2})$, and correlation coefficient vector $(\rho_{12}, \rho_{13}, \rho_{23})$. 
We considered eight different scenarios, representing all combinations of the following parameters:
\begin{align*}
    (\mu_{1}, \mu_{2}, \mu_{3})                                 &= (0.0, 0.0, 0.0), (0.0, -0.1, 0.1), \\
    (\sigma_{1}^{2}, \sigma_{2}^{2}, \sigma_{3}^{2})            &= (1.0, 1.0, 1.0), (1.0, 1.2, 1.4), \\
    (\rho_{12}, \rho_{13}, \rho_{23})   &= (0.2, 0.2, 0.2), (0.2, 0.3, 0.4).
\end{align*}
The simulation procedure is as follows:
\begin{enumerate}[leftmargin=*, label=Step \arabic*.] \itemsep=2mm
    \item Generate 10,000 random numbers from a trivariate normal distribution with parameters\\ $(\mu_{1}, \mu_{2}, \mu_{3})$, $(\sigma_{1}^{2}, \sigma_{2}^{2}, \sigma_{3}^{2})$, and $(\rho_{12}, \rho_{13}, \rho_{23})$.
    
    \item A $4 \times 4 \times 4$ contingency table is formed using cut points for all three variables at $\mu_1$, $\mu_1 \pm 0.6\sigma_1$.
    
    \item Models are fitted for the table.
    
    \item Steps 1, 2, and 3 are repeated 10,000 times.
\end{enumerate}

Table \ref{tab:simulation} lists the empirical power values of the test for each model for eight different configurations.
First, for the configuration in which the symmetric structure holds (i.e., when all means, variances, and correlation coefficients are identical), all the models consistently yield empirical power values below $0.05$. 
This demonstrates that they effectively control the Type I error rate, correctly identify the underlying symmetry, and avoid the spurious detection of asymmetry when none exists.

Examining the results for the nested LS$[f]$ model shows limited explanatory power, primarily accounting for the shifts in means.
Their empirical power rapidly escalates to $1$ when deviations exist in the variances or correlation coefficients. 
Similarly, the ELS$[f]$ model can explain the deviations in both the means and variances.
Additionally, they show a significant increase in empirical power when the correlation coefficients show differences.

Conversely, the GS$[f]$ model consistently demonstrated appropriate explanatory power across all configurations.
Even when other nested models exhibited an empirical power of $1$, our proposed models maintained moderate power. 
This highlights their superior ability to detect various forms of asymmetry generated from a multivariate normal distribution.

Furthermore, these findings hold irrespective of the chosen $f$-function (KL-divergence, Pearson's $\chi^2$-divergence, or Hellinger distance). 
This robust performance across different divergence measures suggests that the proposed GS$[f]$ model offers considerable flexibility in selecting an appropriate divergence, while consistently providing good fits and accurate detection of asymmetric structures.

\begin{table}[htbp]
  \caption{Empirical power of the test for each model for trivariate normal distributions across different parameter configurations.}
  \label{tab:simulation}
  \centering
  \begin{tabular}{ccccccc}
   \toprule
    &  &  & \multicolumn{4}{c}{KL-divergence based Models} \\
   \cmidrule(lr){4-7}
    ($\mu_{1}, \ \mu_{2}, \ \mu_{3}$) & ($\sigma_{1}^{2}, \ \sigma_{2}^{2}, \ \sigma_{3}^{2}$) & ($\rho_{12}, \ \rho_{13}, \ \rho_{23}$) & S & LS & ELS & GS \\
   \midrule
   ($0.0, 0.0, 0.0$) & ($1.0, 1.0, 1.0$) & ($0.2, 0.2, 0.2$) & $0.0479$ & $0.0482$ & $0.0473$ & $0.0495$ \\
   ($0.0, 0.0, 0.0$) & ($1.0, 1.2, 1.4$) & ($0.2, 0.2, 0.2$) & $1$ & $1$ & $0.0549$ & $0.0504$ \\
   ($0.0, 0.0, 0.0$) & ($1.0, 1.0, 1.0$) & ($0.2, 0.3, 0.4$) & $1$ & $1$ & $1$ & $0.1186$ \\
   ($0.0, 0.0, 0.0$) & ($1.0, 1.2, 1.4$) & ($0.2, 0.3, 0.4$) & $1$ & $1$ & $1$ & $0.1432$ \\ \\

   ($0.0, -0.1, 0.1$) & ($1.0, 1.0, 1.0$) & ($0.2, 0.2, 0.2$) & $1$ & $0.0786$ & $0.0799$ & $0.0796$ \\
   ($0.0, -0.1, 0.1$) & ($1.0, 1.2, 1.4$) & ($0.2, 0.2, 0.2$) & $1$ & $1$ & $0.0804$ & $0.0746$ \\
   ($0.0, -0.1, 0.1$) & ($1.0, 1.0, 1.0$) & ($0.2, 0.3, 0.4$) & $1$ & $1$ & $1$ & $0.1751$ \\
   ($0.0, -0.1, 0.1$) & ($1.0, 1.2, 1.4$) & ($0.2, 0.3, 0.4$) & $1$ & $1$ & $1$ & $0.2068$ \\
   \bottomrule \\
   \toprule
    &  &  & \multicolumn{4}{c}{Pearson's $\chi^2$-divergence based Models} \\
   \cmidrule(lr){4-7}
    ($\mu_{1}, \ \mu_{2}, \ \mu_{3}$) & ($\sigma_{1}^{2}, \ \sigma_{2}^{2}, \ \sigma_{3}^{2}$) & ($\rho_{12}, \ \rho_{13}, \ \rho_{23}$) & S & PLS & PELS & PGS \\
   \midrule
   ($0.0, 0.0, 0.0$) & ($1.0, 1.0, 1.0$) & ($0.2, 0.2, 0.2$) & $0.0477$ & $0.0482$ & $0.0476$ & $0.0492$ \\
   ($0.0, 0.0, 0.0$) & ($1.0, 1.2, 1.4$) & ($0.2, 0.2, 0.2$) & $1$ & $1$ & $0.0558$ & $0.0522$ \\
   ($0.0, 0.0, 0.0$) & ($1.0, 1.0, 1.0$) & ($0.2, 0.3, 0.4$) & $1$ & $1$ & $1$ & $0.1488$ \\
   ($0.0, 0.0, 0.0$) & ($1.0, 1.2, 1.4$) & ($0.2, 0.3, 0.4$) & $1$ & $1$ & $1$ & $0.1564$ \\ \\

   ($0.0, -0.1, 0.1$) & ($1.0, 1.0, 1.0$) & ($0.2, 0.2, 0.2$) & $1$ & $0.093$ & $0.0932$ & $0.0792$ \\
   ($0.0, -0.1, 0.1$) & ($1.0, 1.2, 1.4$) & ($0.2, 0.2, 0.2$) & $1$ & $1$ & $0.1297$ & $0.1011$ \\
   ($0.0, -0.1, 0.1$) & ($1.0, 1.0, 1.0$) & ($0.2, 0.3, 0.4$) & $1$ & $1$ & $1$ & $0.4338$ \\
   ($0.0, -0.1, 0.1$) & ($1.0, 1.2, 1.4$) & ($0.2, 0.3, 0.4$) & $1$ & $1$ & $1$ & $0.3567$ \\
   \bottomrule \\
   \toprule
    &  &  & \multicolumn{4}{c}{Hellinger distance based Models} \\
   \cmidrule(lr){4-7}
    ($\mu_{1}, \ \mu_{2}, \ \mu_{3}$) & ($\sigma_{1}^{2}, \ \sigma_{2}^{2}, \ \sigma_{3}^{2}$) & ($\rho_{12}, \ \rho_{13}, \ \rho_{23}$) & S & HLS & HELS & HGS \\
   \midrule
   ($0.0, 0.0, 0.0$) & ($1.0, 1.0, 1.0$) & ($0.2, 0.2, 0.2$) & $0.0478$ & $0.0483$ & $0.0475$ & $0.0492$ \\
   ($0.0, 0.0, 0.0$) & ($1.0, 1.2, 1.4$) & ($0.2, 0.2, 0.2$) & $1$ & $1$ & $0.0555$ & $0.0506$ \\
   ($0.0, 0.0, 0.0$) & ($1.0, 1.0, 1.0$) & ($0.2, 0.3, 0.4$) & $1$ & $1$ & $1$ & $0.1247$ \\
   ($0.0, 0.0, 0.0$) & ($1.0, 1.2, 1.4$) & ($0.2, 0.3, 0.4$) & $1$ & $1$ & $1$ & $0.1471$ \\ \\

   ($0.0, -0.1, 0.1$) & ($1.0, 1.0, 1.0$) & ($0.2, 0.2, 0.2$) & $1$ & $0.0804$ & $0.0803$ & $0.0791$ \\
   ($0.0, -0.1, 0.1$) & ($1.0, 1.2, 1.4$) & ($0.2, 0.2, 0.2$) & $1$ & $1$ & $0.0798$ & $0.0769$ \\
   ($0.0, -0.1, 0.1$) & ($1.0, 1.0, 1.0$) & ($0.2, 0.3, 0.4$) & $1$ & $1$ & $1$ & $0.1784$ \\
   ($0.0, -0.1, 0.1$) & ($1.0, 1.2, 1.4$) & ($0.2, 0.3, 0.4$) & $1$ & $1$ & $1$ & $0.2054$ \\
   \bottomrule
  \end{tabular}
\end{table}

\section{Empirical Analysis} \label{sec:example}
This section assesses the empirical performance of the proposed GS$[f]$ model using a real-world dataset.
However, the analysis is now casts entirely within the Homogeneous Linear Predictor (HLP) framework \citep{lang2005Homogeneous}, which is implemented in \textsf{R} with \texttt{mph.fit}.
The HLP model is defined as $L(\boldsymbol{\pi})=\boldsymbol{X} \boldsymbol{\vartheta}$ with a smooth link function $L$ and design matrix $\boldsymbol{X}$.
In the following, we report the goodness-of-fit statistics for the GS$[f]$ model and each nested alternative.
Herein, we set $\alpha_{3}=\beta_{33}=\beta_{23}=0$ without loss of generality and naturally adopt equally spaced scores $\{u_{i} = i \}$.

Table \ref{table:party} presents data from the American National Election Study (ANES) 2020–2022 Social Media Study \citep{anes2023socialmedia}, a comprehensive panel study of the U.S. adults conducted online.
Respondents were asked, ``Do you think of yourself as closer to the Republican Party or to the Democratic Party?'' with the response options of (1) Closer to the Republican Party, (2) Closer to the Democratic Party, and (3) Neither. 
To ensure ordinal interpretation, we recoded the categories such that ``(3) Neither'' is treated as the middle category, yielding the ordering: (1) Republican, (2) Neither, (3) Democratic.
The survey was administered across three waves: a pre- and post-election study of the 2020 presidential election, followed by the 2022 midterm elections.

Table \ref{table:party_result} summarizes the likelihood ratio chi-square statistics for the various models fitted to the data in Table \ref{table:party}. 
The GS$[f]$ model provided an adequate fit to the data at the 5\% significance level.
However, nested models, such as the ELS and LS, were rejected because of the asymmetrical structures in the correlations.
This asymmetric dependence supports the need to use the GS$[f]$ modeling framework rather than its nested models to accurately capture the transition of political identification over time.

\begin{table}[ht]
  \caption{Political identification across three waves of the ANES 2020--2022 Social Media Study}
  \label{table:party}
  \centering
  \begin{tabular}{llccc}
   \toprule
     & & \multicolumn{3}{c}{Post-election (2022)} \\
   \cmidrule(lr){3-5}
     Pre-election (2020) & Post-election (2020) & (1) Republican & (2) Neither & (3) Democratic \\
   \midrule 
                    & (1) Republican  & 240 & 32 & 8 \\
   (1) Republican   & (2) Neither     &  11 & 23 & 5 \\
                    & (3) Democratic  &   0 &  2 & 4 \\
   \midrule
                    & (1) Republican  & 20 &  22 &  4 \\
    (2) Neither     & (2) Neither     & 18 & 237 & 28 \\
                    & (3) Democratic  &  1 &  24 & 29 \\
   \midrule
                    & (1) Republican  & 4 &  0 & 5 \\
   (3) Democratic   & (2) Neither     & 0 & 28 & 16 \\
                    & (3) Democratic  & 7 & 36 & 323 \\
   \bottomrule
  \end{tabular}
\end{table}

\setlength{\tabcolsep}{5pt}
\begin{table}[ht]
\caption{Likelihood ratio chi-square values $G^2$ for models applied to the data in Table \ref{table:party}.}
\label{table:party_result}
\centering
    \begin{tabular}{ccccccccccc}
     \toprule
     & & & & & & \multicolumn{3}{c}{Proposed Models} & \multicolumn{2}{c}{Nested Models} \\
     \cmidrule(lr){7-9} \cmidrule(lr){10-11}
     Models     & S & ME$_2$ & ME & VE & CE & GS & PGS & HGS & ELS & LS \\
     \midrule
     \textit{df}& $17$ & $6$ & $2$ & $2$ & $2$ & $11$ & $11$ & $11$ & $13$ & $15$ \\
     $G^2$      & $45.3$ & $31.6$ & $3.34$ & $9.89$ & $17.4$ & $15.5$ & $13.7$ & $16.0$ & $33.0$ & $41.5$ \\
     $p$-value  & $<0.001$ & $<0.001$ & $0.188$ & $0.007$ & $<0.001$ & $0.162$ & $0.249$ & $0.143$ & $0.002$ & $<0.001$ \\
    \bottomrule
    \end{tabular}
\end{table}
\setlength{\tabcolsep}{6pt}

Table \ref{table:party_potential} presents the plug-in estimates of the potential parameters for the GS, PGS, and HGS models. 
These estimates were obtained by substituting the MLEs of the model parameters $\boldsymbol{\vartheta}$ into the model-specific expressions for each potential parameter $\theta_{\boldsymbol{i}}$. 
Despite the differences in the selected $f$-function, the relative ordering of the estimates within each group remained consistent across the models.
However, the HGS model inverts the order of the potential parameters, although the overall ranking remains aligned with that of the GS and PGS models (see Section \ref{subsec:specialcase}).

Group 1 exhibits particularly large discrepancies across the models, indicating a strong deviation from complete symmetry.
This is evident in the following relationships, which are higher than all the deviations in the other groups:
\begin{equation}
    \frac{\hat \pi_{(1,1,3)}^{c}}{\hat \pi_{(1,3,1)}^{c}} = 8.27, \quad
    \hat \pi_{(1,1,3)}^{c} - \hat \pi_{(1,3,1)}^{c} = 0.71, \quad
    (\hat \pi_{(1,1,3)}^{c})^{-\frac{1}{2}} - (\hat \pi_{(1,3,1)}^{c})^{-\frac{1}{2}} = - 1.88.
\end{equation}
These quantities can be interpreted in light of the benchmark values under complete symmetry: the ratio-based measure (KL-divergence case) is equal to $1$, the difference-based measure (Pearson’s $\chi^2$ divergence case) is equal to $0$, and the Hellinger-type measure is also equal to $0$. 
Thus, departures from these values indicate the degree of asymmetry present within specific groups.
The ratio- and difference-based formulations offer intuitive interpretations in terms of multiplicative and additive deviations from symmetry, respectively. 
In contrast, the Hellinger-type formulation is less directly interpretable but serves to illustrate how the GS[$f$] framework can express a wider class of $f$-divergence–based relationships. 
To our knowledge, this perspective has not been emphasized in previous studies, and it highlights the added interpretability that the GS[$f$] framework can bring to the analysis of asymmetry in contingency tables.

Although conventional models are limited in capturing symmetric or ratio-based structures, the GS[$f$] framework flexibly accommodates a broader range of asymmetric relationships among conditional probabilities. 
These differences illustrate that each model embodies a distinct concept of asymmetry. 
Consequently, researchers should select the $f$-function that best aligns with their analytical objectives, considering whether ratio-based, difference-based, or other forms of structural deviation are most substantively meaningful in the given application.

\begin{table}[ht]
\caption{Plug-in estimates of potential parameters $\theta_{\boldsymbol{i}}$ for the GS, PGS, and HGS models computed using the MLEs of the model parameters $\boldsymbol{\vartheta}$ fitted to the data in Table \ref{table:party}.}
\label{table:party_potential}
\centering
    \begin{tabular}{lcccc}
        \toprule
        & & \multicolumn{3}{c}{Models} \\
        \cmidrule(lr){3-5}
        Groups & \multirow{2}{*}[3mm]{\parbox{4cm}{\centering Plug-in estimates \\ of potential parameters}} & GS & PGS & HGS \\
        \midrule
        Group 1 & $\theta_{(1,1,3)}$ & $0.0161$   & $-1.3085$    & $3.7307$ \\
                & $\theta_{(3,1,1)}$ & $0.0057$            & $-1.7262$             & $4.5959$          \\
                & $\theta_{(1,3,1)}$ & $0.0019$            & $-2.0173$             & $5.6129$          \\
        \\
        Group 2 & $\theta_{(1,1,2)}$ & $0.0249$   & $-1.1937$    & $3.3065$ \\
                & $\theta_{(2,1,1)}$ & $0.0139$            & $-1.4227$             & $3.7881$          \\
                & $\theta_{(1,2,1)}$ & $0.0078$            & $-1.5785$             & $4.3284$          \\
        \\
        Group 3 & $\theta_{(3,3,1)}$ & $0.0011$   & $-2.1269$    & $6.2432$ \\
                & $\theta_{(1,3,3)}$ & $0.0008$            & $-2.2469$             & $6.4613$          \\
                & $\theta_{(3,1,3)}$ & $0.0004$            & $-2.4150$             & $7.1410$          \\
        \\
        Group 4 & $\theta_{(3,3,2)}$ & $0.0003$   & $-2.4713$    & $7.5158$ \\
                & $\theta_{(2,3,3)}$ & $0.0002$            & $-2.5515$             & $7.6739$          \\
                & $\theta_{(3,2,3)}$ & $0.0002$            & $-2.6458$             & $8.0455$          \\
        \\
        Group 5 & $\theta_{(1,2,3)}$ & $0.0033$   & $-0.9040$    & $7.3211$ \\
                & $\theta_{(2,1,3)}$ & $0.0024$            & $-0.9409$             & $7.7568$          \\
                & $\theta_{(3,2,1)}$ & $0.0022$            & $-0.9785$             & $7.7786$          \\
                & $\theta_{(3,1,2)}$ & $0.0015$            & $-1.0353$             & $8.2992$          \\
                & $\theta_{(2,3,1)}$ & $0.0014$            & $-1.0461$             & $8.4529$          \\
                & $\theta_{(1,3,2)}$ & $0.0013$            & $-1.0661$             & $8.5378$          \\                
        \\
        Group 6 & $\theta_{(2,1,2)}$ & $0.0058$   & $-1.6523$    & $4.6365$ \\
                & $\theta_{(1,2,2)}$ & $0.0051$            & $-1.6933$             & $4.7526$          \\
                & $\theta_{(2,2,1)}$ & $0.0039$            & $-1.7878$             & $4.9634$          \\                
        \\
        Group 7 & $\theta_{(2,2,3)}$ & $0.0007$   & $-2.2471$    & $6.6602$ \\
                & $\theta_{(3,2,2)}$ & $0.0006$            & $-2.3013$             & $6.7729$          \\
                & $\theta_{(2,3,2)}$ & $0.0006$            & $-2.3219$             & $6.8255$          \\        
        \bottomrule
    \end{tabular}
\end{table}

\section{Conclusion} \label{sec:conclusion}
This study proposes a novel framework for modeling multivariate categorical data by incorporating $f$-divergence.
The proposed model extends the existing approaches by allowing for a more general class of divergence measures in the maximum entropy principle, offering a theoretical basis for understanding traditional models.
Furthermore, it captures diverse probabilistic structures by selecting $f$-functions, thereby addressing the practical need to tailor models to specific data and research objectives.
Additionally, the potential parameters used to assess the effects of individual cells enhance interpretability.

Future studies should extend this framework in two main directions. 
The first objective is to construct a theoretical framework for maximum-entropy models under fixed general constraints.
Second, we explore modeling strategies based on alternative divergence measures such as the Bregman divergence, which plays a central role in information geometry owing to its strong geometric interpretation \citep{csiszar2004Information, murata2004Information, amari2010Information}.
Together, these avenues of investigation aim to deepen the theoretical understanding and broaden the applicability of the model across diverse research contexts.

\section*{Acknowledgements}
This work was supported by JST SPRING (Grant Number JPMJSP2151) and JSPS KAKENHI (Grant Number JP20K03756).
The authors sincerely thank Professor Joseph B. Lang for generously sharing his \textsf{R} code, which was essential for the analysis in this study, and Dr. Kengo Fujisawa for his valuable advice.

\section*{Data Availability}
The data analyzed in this study are publicly available from the American National Election Study (ANES) at \url{https://electionstudies.org/}.

\section*{Code Availability}
All the \textsf{R} scripts required to reproduce the analyses in this study are openly available at \url{https://github.com/h-okahara/GS_f}.
The analyses also rely on function \texttt{mph.fit} contained in the file \texttt{mph.Rcode.R}, which is not included in the GitHub repository.
Researchers who wish to use \texttt{mph.fit} can contact Professor Joseph B. Lang directly to obtain this file.

\section*{Appendix A Notation Summary} \label{appendix-notation}
Table \ref{table:notation} summarizes the notations used in this study.
\begin{table}[htbp]
  \caption{Notation reference}
  \label{table:notation}
  \centering
  \begin{tabular}{cl}
    \toprule
    \textbf{Symbol}                 & \multicolumn{1}{c}{\textbf{Definition}} \\
    \midrule
    $V$                             & The indices set, usually $\{1,\dots,T\}$ \\
    \addlinespace[1pt]
    $\mathcal{I}$                   & Levels of any variable, by default $\{1,\dots,r\}$    \\
    \addlinespace[1pt]
    $\boldsymbol{i}$                & $(i_1,\dots,i_T)$ with $i_j \in \mathcal{I}$, generic notation to denote a cell   \\
    \addlinespace[1pt]
    $\mathcal{I}_{V}$               & $\mathcal{I}^{|V|} = \mathcal{I} \times \cdots \times \mathcal{I}$, the collection of all cells   \\
    \addlinespace[1pt]
    $D(\boldsymbol{i})$             & This set contains all permutations of the components of $\boldsymbol{i}$  \\
    \addlinespace[1pt]
    $\boldsymbol{\pi}$              & Joint \textrm{p.m.f} of $(X_1,\dots,X_T)$, $\boldsymbol{\pi} = \{\pi_{\boldsymbol{i}}\}_{\boldsymbol{i} \in \mathcal{I}_{V}}$   \\
    \addlinespace[1pt]
    $\boldsymbol{\pi}^{S}$          & Joint \textrm{p.m.f} of the complete symmetry, $\boldsymbol{\pi}^{S} = \{\pi_{\boldsymbol{i}}^{S}\}_{\boldsymbol{i} \in \mathcal{I}_{V}}$ \\ 
    \addlinespace[1pt]
    $\mathcal{I}_{V}^{S}$           & Subset of $\mathcal{I}_{V}$ that imposes restrictions on the order of its elements    \\
    \addlinespace[1pt]
    $\mathcal{V}_{2}$               & Subset of $V \times V$, the collection of all index of off-diagonal of $\boldsymbol{B}$   \\
    \addlinespace[1pt]
    $\preceq$                       & The lexicographical ordering defined on finite vector space $\mathbb{R}^{n}$  \\
    \addlinespace[1pt]
    $\boldsymbol{\iota}$            & $(\boldsymbol{\iota}_{1},\dots,\boldsymbol{\iota}_{|\mathcal{I}_{V}|})$ is elements of $\mathcal{I}_{V}$ arranged in lexicographic order \\
    \addlinespace[1pt]
    $\boldsymbol{\sigma}$           & $(\boldsymbol{\sigma}_{1},\dots,\boldsymbol{\sigma}_{|\mathcal{I}_{V}^{S}|})$ is elements of $\mathcal{I}_{V}^{S}$ arranged in lexicographic order   \\
    \addlinespace[1pt]
    $\boldsymbol{\nu}$              & $(\boldsymbol{\nu}_{1},\dots,\boldsymbol{\nu}_{|\mathcal{V}_{2}|}\}$ is elements of $\mathcal{V}_{2}$ arranged in lexicographic order \\
    \addlinespace[1pt]
    $\boldsymbol{s}_{h}$            & Score vector for variable $h$, defined as $\boldsymbol{1}_{r^{h-1}} \otimes (u_{1}, \dots, u_{r})^\top \otimes \boldsymbol{1}_{r^{T-h}}$ \\
    \addlinespace[1pt]
    $\boldsymbol{\Delta}_{1}^{(h)}$ & First-order difference for $h \in \{1,\dots,T-1\}$ \\
    \addlinespace[1pt]
    $\boldsymbol{\Delta}_{2}^{(h)}$ & Second-order difference for $h \in \{1,\dots,T-1\}$ \\
    \addlinespace[1pt]
    $\boldsymbol{\Delta}_{\mathcal{V}_2}^{(k)}$ & Second-order interaction difference for pair $\boldsymbol{\nu}_{k} = (s,t) \in \mathcal{V}_2$ \\
    \addlinespace[1pt]
    \bottomrule
  \end{tabular}
\end{table}

\section*{Appendix B Proofs} \label{appendix-proofs}

\subsection*{Proof of Theorem \ref{thm-2}}
If the S model holds, then the GS$[f]$ and ME$_2$ models also hold. 
Assuming that these models hold, we show that the S model holds.
Let $\{\hat \pi_{\boldsymbol{i}}\}$ denote the cell probabilities satisfy the structures of GS$[f]$ and ME$_2$ models. 
$\{\hat \pi_{\boldsymbol{i}}\}$ satisfies the GS$[f]$ model.
\begin{equation*}
    F \biggl( \frac{\hat \pi_{\boldsymbol{i}} }{\hat \pi_{\boldsymbol{i}}^{S}} \biggr) 
    = \boldsymbol{u}_{\boldsymbol{i}}^{\top} \boldsymbol{\alpha} 
    + \boldsymbol{u}_{\boldsymbol{i}}^{\top} \boldsymbol{B} \boldsymbol{u}_{\boldsymbol{i}}
    + \gamma_{\boldsymbol{i}}, 
    \quad \boldsymbol{i} \in \mathcal{I}_{V},
\end{equation*}
where $\boldsymbol{B}$ is a symmetric matrix and $\gamma_{\boldsymbol{i}} = \gamma_{\boldsymbol{j}}$ for $\boldsymbol{j} \in D(\boldsymbol{i})$. 
Subsequently, for $\boldsymbol{i} \in \mathcal{I}_{V}$ and $\boldsymbol{j} \in D(\boldsymbol{i})$,
\begin{equation} \label{eq-th3.1}
    F \biggl( \frac{\hat \pi_{\boldsymbol{i}}}{\hat \pi_{\boldsymbol{i}}^{S}} \biggr) - F \biggl( \frac{\hat \pi_{\boldsymbol{j}}}{\hat \pi_{\boldsymbol{j}}^{S}} \biggr) 
    = (\boldsymbol{u}_{\boldsymbol{i}} - \boldsymbol{u}_{\boldsymbol{j}})^{\top} \boldsymbol{\alpha}
    + (\boldsymbol{u}_{\boldsymbol{i}} - \boldsymbol{u}_{\boldsymbol{j}})^{\top} \boldsymbol{B} (\boldsymbol{u}_{\boldsymbol{i}} + \boldsymbol{u}_{\boldsymbol{j}}).
\end{equation}
The sum of the right-hand side of equation \eqref{eq-th3.1} multiplied by $\hat \pi_{\boldsymbol{i}}$ is zero for each term, for the following reasons:
For term $\boldsymbol{\alpha} = (\alpha_{1}, \dots, \alpha_{T})^{\top}$,
\begin{equation*}
     \sum_{\boldsymbol{i} \in \mathcal{I}_{V}} \hat \pi_{\boldsymbol{i}} (\boldsymbol{u}_{\boldsymbol{i}} - \boldsymbol{u}_{\boldsymbol{j}})^{\top} \boldsymbol{\alpha}
     = \sum_{s \in V} \alpha_{s} \sum_{\boldsymbol{i} \in \mathcal{I}_{V}} (u_{i_{s}} - u_{j_{s}}) \hat \pi_{\boldsymbol{i}} = 0,
\end{equation*}
because $\hat \pi_{\boldsymbol{i}}$ satisfies the ME structure.
Because the VE and CE models hold, this immediately implies that the covariances are equal.
Therefore, for $\boldsymbol{B} = \{ \beta_{st} \}$, the following equation holds.
\begin{equation*}
    \sum_{\boldsymbol{i} \in \mathcal{I}_{V}} \hat \pi_{\boldsymbol{i}} (\boldsymbol{u}_{\boldsymbol{i}} - \boldsymbol{u}_{\boldsymbol{j}})^{\top} \boldsymbol{B} (\boldsymbol{u}_{\boldsymbol{i}} + \boldsymbol{u}_{\boldsymbol{j}})
    = \sum_{s,t \in V} \beta_{st} \sum_{\boldsymbol{i} \in \mathcal{I}_{V}} (u_{i_{s}} u_{i_{t}} - u_{j_{s}} u_{j_{t}}) \hat \pi_{\boldsymbol{i}} = 0,
\end{equation*}
because $\hat \pi_{\boldsymbol{i}}$ satisfies the VE and CE model structures.
Therefore,
\begin{equation} \label{eq-th3.2}
    \sum_{\boldsymbol{i} \in \mathcal{I}_{V}} \hat \pi_{\boldsymbol{i}} \left( F \biggl( \frac{\hat \pi_{\boldsymbol{i}} }{\hat \pi_{\boldsymbol{i}}^{S}} \biggr) - F \biggl( \frac{\hat \pi_{\boldsymbol{j}} }{\hat \pi_{\boldsymbol{j}}^{S}} \biggr) \right) = 0.
\end{equation}
Notably, the ME$_2$ model is equivalent to the condition under which the ME, VE, and CE models hold simultaneously, capturing the equality of marginal means, variances, and covariances.

Here, we define the following set:
\begin{align*}
    D_{n}(\boldsymbol{i}) &= D_{n}(i_1,\dots,i_T) \\
    &= \{ \boldsymbol{j} \in \mathcal{I}_{V} \ | \ \boldsymbol{j} \ \text{is any permutation of the} \ n \ \text{elements of} \ \boldsymbol{i} \},
\end{align*}
for a given $n \in V$, 
and function
\begin{equation} \label{eq-G_pc}
    G(\boldsymbol{i}, \boldsymbol{j}) 
    = F \biggl( \frac{\hat \pi_{\boldsymbol{i}} }{\hat \pi_{\boldsymbol{i}}^{S}} \biggr) 
    - F \biggl( \frac{\hat \pi_{\boldsymbol{j}} }{\hat \pi_{\boldsymbol{j}}^{S}} \biggr).
\end{equation}
Function $F$ increases monotonically, $G(\boldsymbol{i}, \boldsymbol{j})=0$ is satisfied only when $\hat \pi_{\boldsymbol{i}}=\hat \pi_{\boldsymbol{j}}$.
Note that if $\boldsymbol{j} \in D(\boldsymbol{i})$, then $\pi_{\boldsymbol{i}}^S = \pi_{\boldsymbol{j}}^S$.
We then prove that the following Proposition holds:
\begin{equation} \label{prop-th3.1}
    \forall n \in V, \ \forall \boldsymbol{j} \in D_{n}(\boldsymbol{i}), \ \sum_{\boldsymbol{i} \in \mathcal{I}_{V}} \hat \pi_{\boldsymbol{i}} G(\boldsymbol{i}, \boldsymbol{j}) = 0
    \ \Rightarrow \ 
    \hat \pi_{\boldsymbol{i}} = \hat \pi_{\boldsymbol{j}}.
\end{equation}
First, we consider $\boldsymbol{j} \in D_{2}(\boldsymbol{i})$ without loss of generality, permuting $i_{1}$ and $i_{2}$ in $\boldsymbol{i}$. 
The following equivalence holds.
\begin{equation}
    \sum_{\boldsymbol{i} \in \mathcal{I}_{V}} \hat \pi_{\boldsymbol{i}} G(\boldsymbol{i}, \boldsymbol{j}) = 0
    \ \Leftrightarrow \
    \sum_{\substack{\boldsymbol{i} \in \mathcal{I}_{V} \\ i_{1} < i_{2}}} (\hat \pi_{\boldsymbol{i}} - \hat \pi_{\boldsymbol{j}} ) G(\boldsymbol{i}, \boldsymbol{j}) = 0, 
\end{equation}
we can easily obtain $\hat \pi_{\boldsymbol{i}} = \hat \pi_{\boldsymbol{j}}$. Thus, the Proposition \eqref{prop-th3.1} holds when $n=2$. 
Assuming that proposition \eqref{prop-th3.1} holds for the case $n=k-1$, we consider the case $n=k$. 
Without a loss of generality, let $\boldsymbol{j} \in D_{k}(\boldsymbol{i})$ be a permutation of $(i_{1},\dots,i_{k})$ in $\boldsymbol{i}$. 
If $j_{k}=i_{k}$, then $\boldsymbol{j}$ belongs to $D_{k-1}(\boldsymbol{i})$; hence, by the induction hypothesis, proposition \eqref{prop-th3.1} holds. 
If $j_{k} \neq i_{k}$, then interchanging $i_k$ and $j_k$ in $\boldsymbol{j}$ yields $\boldsymbol{k} \in D_{k-1}(\boldsymbol{i})$. 
Based on the induction hypothesis, we propose the following hypothesis:
\begin{align}
    \sum_{\boldsymbol{i} \in \mathcal{I}_{V}} \hat \pi_{\boldsymbol{i}} G(\boldsymbol{i}, \boldsymbol{j}) = 0
    & \Leftrightarrow \sum_{\boldsymbol{i} \in \mathcal{I}_{V}} \hat \pi_{\boldsymbol{i}} G(\boldsymbol{i}, \boldsymbol{k}) + \sum_{\boldsymbol{i} \in \mathcal{I}_{V}} \hat \pi_{\boldsymbol{i}} G(\boldsymbol{k}, \boldsymbol{j}) = 0 \\
    & \Rightarrow \hat \pi_{\boldsymbol{i}} = \hat \pi_{\boldsymbol{k}} \land \sum_{\boldsymbol{k} \in \mathcal{I}_{V}} \hat \pi_{\boldsymbol{k}} G(\boldsymbol{k}, \boldsymbol{j}) = 0 \\
    & \Rightarrow \hat \pi_{\boldsymbol{i}} = \hat \pi_{\boldsymbol{k}} = \hat \pi_{\boldsymbol{j}}.
\end{align}
Therefore, Proposition \eqref{prop-th3.1} holds; namely, $\hat \pi_{\boldsymbol{i}} = \hat \pi_{\boldsymbol{j}}$ for any permutation of $\boldsymbol{i}$, and $\{ \hat \pi_{\boldsymbol{i}} \}$ satisfy the structure of complete symmetry. \\
\qed

\subsection*{Proof of Equation \eqref{defn-gs-f-linear}.}
We consider the following optimization problem posed on the same feasible set as in the proof of Theorem \ref{thm-1}:
\begin{align}
    \text{Minimize}     
    &\qquad  D_{f}(\boldsymbol{\pi} \parallel \boldsymbol{\pi}^{S}) \\
    \text{subject to}   
    &\quad \sum_{\boldsymbol{j} \in D(\boldsymbol{i})} \pi_{\boldsymbol{j}} = v_{\boldsymbol{i}}^{S}, \quad \boldsymbol{i} \in \mathcal{I}_{V}, \\
    &\quad \left(\boldsymbol{\Delta}_{1}^{(h)} \right)^{\top} \boldsymbol{\pi} = 0,
    \quad \left(\boldsymbol{\Delta}_{2}^{(h)}\right)^{\top} \boldsymbol{\pi} = 0, 
    \quad h = 1,\dots, T-1, \\
    &\quad \left( \boldsymbol{\Delta}_{\mathcal{V}_{2}}^{(k)} \right)^{\top} \boldsymbol{\pi} = 0,  
    \quad k = 1, \dots, |\mathcal{V}_{2}|.
\end{align}
This optimization problem can be solved using Lagrange multipliers as follows:
\begin{align*}
    L(\{\pi_{\boldsymbol{i}} \}) = D_{f}( \boldsymbol{\pi} \parallel \boldsymbol{\pi}^{S}) 
    &+ \sum_{h=1}^{T-1} \left\{ \lambda_{1}^{(h)} \left(\boldsymbol{\Delta}_{1}^{(h)} \right)^{\top} 
    + \lambda_{2}^{(h)} \left(\boldsymbol{\Delta}_{2}^{(h)} \right)^{\top} \right\} \boldsymbol{\pi}
    \\
    &+ \sum_{k=1}^{|\mathcal{V}_{2}|} \lambda_{\mathcal{V}_{2}}^{(k)} \left( \boldsymbol{\Delta}_{\mathcal{V}_{2}}^{(k)} \right)^{\top} \boldsymbol{\pi} 
    + \sum_{\boldsymbol{i} \in \mathcal{I}_{V}} \eta_{\boldsymbol{i}} \Biggl( \sum_{\boldsymbol{j} \in D(\boldsymbol{i})} \pi_{\boldsymbol{j}} 
    - v_{\boldsymbol{i}}^{S} \Biggr).
\end{align*}
Equating the partial derivative of $L(\{\pi_{\boldsymbol{i}} \})$ to $0$ with respect to $\pi_{\boldsymbol{i}}$ yields:
\begin{align}
    0 = f^{\prime} \left( \frac{\pi_{\boldsymbol{i}}}{\pi_{\boldsymbol{i}}^{S}} \right) 
    &+ \sum_{h=1}^{T-1} \left\{ \lambda_{1}^{(h)} (u_{i_{h}} - u_{i_{h+1}})
    + \lambda_{2}^{(h)} (u_{i_{h}}^{2} - u_{i_{h+1}}^{2}) \right\} \\
    &+ \sum_{k=1}^{|\mathcal{V}_{2}|} \lambda_{\mathcal{V}_{2}}^{(k)} ( u_{i_{s}} u_{i_{t}} - u_{i_{s^{+}}} u_{i_{t^{+}}})
    + \sum_{\boldsymbol{j} \in D(\boldsymbol{i})} \eta_{\boldsymbol{j}},
\end{align}
where $(s^{+},t^{+})$ denotes the pair following $(s,t) \in \mathcal{V}_2$.
Herein, let $-\lambda_{1}^{(h)} = \alpha_{h}^{\prime}$, $-\lambda_{2}^{(h)} = \beta_{hh}^{\prime}$, $-\lambda_{\mathcal{V}_{2}}^{(k)} = \beta_{\boldsymbol{\nu}_{k}}^{\prime}$ and $-\sum_{\boldsymbol{j} \in D(\boldsymbol{i})} \eta_{\boldsymbol{j}} = \gamma_{\boldsymbol{i}}^{\prime}$.
Subsequently, we obtain the following form:
\begin{equation}
    F\biggl( \frac{\boldsymbol{\pi}}{\boldsymbol{\pi}^{S}} \biggr) =
     \sum_{h=1}^{T-1} \left( \boldsymbol{\Delta}_{1}^{(h)} \alpha_{h}^{\prime}
    + \boldsymbol{\Delta}_{2}^{(h)} \beta_{hh}^{\prime}
    \right)
    + \sum_{k=1}^{|\mathcal{V}_{2}|} \boldsymbol{\Delta}_{\mathcal{V}_{2}}^{(k)} \beta_{\boldsymbol{\nu}_{k}}^{\prime}
    + \boldsymbol{X}^S \boldsymbol{\gamma}^{\prime}
    = \boldsymbol{X} \boldsymbol{\vartheta},
\end{equation}
where 
\begin{equation}
  \boldsymbol{X} =
  \Bigl(
    \boldsymbol{\Delta}_{1}^{(1)}, \dots, \boldsymbol{\Delta}_{1}^{(T-1)},
    \boldsymbol{\Delta}_{2}^{(1)}, \dots, \boldsymbol{\Delta}_{2}^{(T-1)},
    \boldsymbol{\Delta}_{\mathcal{V}_2}^{(1)}, \dots, \boldsymbol{\Delta}_{\mathcal{V}_2}^{(|\mathcal{V}_{2}|)},
    \boldsymbol{X}^{S}
  \Bigr).
\end{equation}
\qed

\subsection*{Proof of Equation \eqref{eq-nondiagonal}.}
Matrix $\boldsymbol{F}$ and its elements $F_{ij}$ are
\begin{equation}
    \boldsymbol{F} = \frac{\partial}{\partial \boldsymbol{\pi}^\top} F \left( \frac{\boldsymbol{\pi}}{\boldsymbol{\pi}^S} \right),
\end{equation}
and
\begin{equation}
F_{ij} = 
\left\{
\begin{array}{cl}
    \dfrac{1}{\pi_{\boldsymbol{\iota}_{i}}^S} \left(1 - \dfrac{\pi_{\boldsymbol{\iota}_{i}}}{|D(\boldsymbol{\iota}_{i})| \pi_{\boldsymbol{\iota}_{i}}^S} \right) f'' \biggl( \frac{\pi_{\boldsymbol{\iota}_{i}}}{\pi_{\boldsymbol{\iota}_{i}}^S} \biggr) & j=i, \\
    - \dfrac{\pi_{\boldsymbol{\iota}_{i}}}{|D(\boldsymbol{\iota}_{i})| (\pi_{\boldsymbol{\iota}_{i}}^S)^2} f'' \biggl( \dfrac{\pi_{\boldsymbol{\iota}_{i}}}{\pi_{\boldsymbol{\iota}_{i}}^S} \biggr) & j \neq i \land \boldsymbol{\iota}_{j} \in D(\boldsymbol{\iota}_{i}), \\
    0       & \textit{otherwise}.
\end{array}
\right.
\end{equation}
Thus, we obtain
\begin{equation}
    \boldsymbol{H}_{1} (\boldsymbol{\pi}) \boldsymbol{\pi} = \boldsymbol{0}_{d_1},
\end{equation}
and under hypothesis $\boldsymbol{h}_{3} (\boldsymbol{\pi}) = \boldsymbol{0}_{d_3}$; that is, under the S model, we see
\begin{equation}
    \boldsymbol{H}_{1} (\boldsymbol{\pi}) \text{diag}(\boldsymbol{\pi})
    = c \ \boldsymbol{U}^\top \bigl(\boldsymbol{I} - \boldsymbol{J} \bigr),
\end{equation}
where $c = f''(1)$ and the elements of $\boldsymbol{J} = \{J_{ij}\}$ are
\begin{equation}
J_{ij} = 
\left\{
\begin{array}{cl}
    \dfrac{1}{|D(\boldsymbol{\iota}_{j})|} & \boldsymbol{\iota}_{j} \in D(\boldsymbol{\iota}_{i}), \\
    0       & \textit{otherwise}.
\end{array}
\right.
\end{equation}
It is straightforward to verify that $\boldsymbol{J}$ belongs to the space $S(\boldsymbol{X})$. 
Therefore, noting that $\boldsymbol{U}^\top \boldsymbol{J} = \boldsymbol{O}_{d_{1}, r^T}$ and $\boldsymbol{M} \boldsymbol{U} = \boldsymbol{O}_{d_{2}, d_{1}}$, we obtained that under the S model 
\begin{equation}
    \boldsymbol{H}_{1} (\boldsymbol{\pi}) \boldsymbol{\Sigma}(\boldsymbol{\pi}) \boldsymbol{H}_{2}^\top (\boldsymbol{\pi}) = \boldsymbol{O}_{d_{1}, d_{2}}.
\end{equation}
\qed

\bibliographystyle{apalike}
\bibliography{main.bib}

\end{document}